\documentclass[11pt]{article}
\usepackage[margin=1in]{geometry} 
\usepackage{amssymb,amsfonts,amsmath,bbm,mathrsfs,stmaryrd}
\usepackage{xcolor}
\usepackage{url}

\usepackage{enumerate}

\usepackage[colorlinks,
             linkcolor=black!75!red,
             citecolor=blue,
             pdftitle={},
             pdfproducer={pdfLaTeX},
             pdfpagemode=None,
             bookmarksopen=true
             bookmarksnumbered=true]{hyperref}

\usepackage[amsmath,thmmarks,hyperref]{ntheorem}
\usepackage{indentfirst}

\theoremstyle{plain}

\newtheorem{lemma}{Lemma}[section]
\newtheorem{proposition}[lemma]{Proposition}
\newtheorem{corollary}[lemma]{Corollary}
\newtheorem{theorem}[lemma]{Theorem}

\newtheorem{question}[lemma]{Question}

\theoremstyle{nonumberplain}

\theoremstyle{plain}
\theorembodyfont{\upshape}
\theoremsymbol{\ensuremath{\blacklozenge}}

\newtheorem{definition}[lemma]{Definition}
\newtheorem{example}[lemma]{Example}
\newtheorem{remark}[lemma]{Remark}

\theoremstyle{nonumberplain}
\theoremsymbol{\ensuremath{\blacksquare}}

\newtheorem{proof}{Proof}

\newcommand{\set}[1]{\left\{ #1\right\}}
\renewcommand{\vec}[1]{\textbf{#1}}

\numberwithin{equation}{section}
\renewcommand{\theequation}{\thesection-\arabic{equation}}

\usepackage{cleveref}

\creflabelformat{enumi}{#2(#1)#3}

\crefname{section}{Section}{Sections}
\crefformat{section}{#2Section~#1#3} 
\Crefformat{section}{#2Section~#1#3} 

\crefname{subsection}{\S}{\S\S}
\crefformat{subsection}{#2\S#1#3} 
\Crefformat{subsection}{#2\S#1#3} 

\crefname{definition}{definition}{definitions}
\crefformat{definition}{#2definition~#1#3} 
\Crefformat{definition}{#2Definition~#1#3} 

\crefname{ex}{example}{examples}
\crefformat{example}{#2example~#1#3} 
\Crefformat{example}{#2Example~#1#3} 

\crefname{remark}{remark}{remarks}
\crefformat{remark}{#2remark~#1#3} 
\Crefformat{remark}{#2Remark~#1#3} 

\crefname{convention}{convention}{conventions}
\crefformat{convention}{#2convention~#1#3} 
\Crefformat{convention}{#2Convention~#1#3} 

\crefname{lemma}{lemma}{lemmas}
\crefformat{lemma}{#2lemma~#1#3} 
\Crefformat{lemma}{#2Lemma~#1#3} 

\crefname{proposition}{proposition}{propositions}
\crefformat{proposition}{#2proposition~#1#3} 
\Crefformat{proposition}{#2Proposition~#1#3} 

\crefname{corollary}{corollary}{corollaries}
\crefformat{corollary}{#2corollary~#1#3} 
\Crefformat{corollary}{#2Corollary~#1#3} 

\crefname{theorem}{theorem}{theorems}
\crefformat{theorem}{#2theorem~#1#3} 
\Crefformat{theorem}{#2Theorem~#1#3} 

\crefname{assumption}{assumption}{Assumptions}
\crefformat{assumption}{#2assumption~#1#3} 
\Crefformat{assumption}{#2Assumption~#1#3} 

\crefname{conjecture}{conjecture}{Conjectures}
\crefformat{conjecture}{#2conjecture~#1#3} 
\Crefformat{conjecture}{#2Conjecture~#1#3} 

\crefname{equation}{}{}
\crefformat{equation}{(#2#1#3)} 
\Crefformat{equation}{(#2#1#3)}

\newcommand\numberthis{\addtocounter{equation}{1}\tag{\theequation}}

\begin{document}

\title{Affine equivalence for quadratic rotation symmetric Boolean functions}
\author{Alexandru Chirvasitu  \footnote{University at Buffalo, Buffalo, NY, USA;  e-mail: achirvas@buffalo.edu}, Thomas W. Cusick \footnote{University at Buffalo, Buffalo, NY, USA;  e-mail: cusick@buffalo.edu}}
\date{}

\maketitle
\begin{abstract}
  Let $f_n(x_0, x_1, \ldots, x_{n-1})$ denote the algebraic normal form (polynomial form) of a rotation symmetric (RS) Boolean function of degree $d$ in $n \geq d$ variables and let $wt(f_n)$ denote the Hamming weight of this function.  Let $(0, a_1, \ldots, a_{d-1})_n$ denote the function $f_n$ of degree $d$ in $n$ variables generated by the monomial $x_0x_{a_1} \cdots x_{a_{d-1}}.$ Such a function $f_n$ is called monomial rotation symmetric (MRS). It was proved in a $2012$ paper that for any MRS $f_n$ with $d=3,$ the sequence of weights $\{w_k = wt(f_k):~k = 3, 4, \ldots\}$ satisfies a homogeneous linear recursion with integer coefficients.  This result was gradually generalized in the following years, culminating around $2016$ with the proof that such recursions exist for any rotation symmetric function $f_n.$ Recursions for quadratic RS functions were not explicitly considered, since a $2009$ paper had already shown that the quadratic weights themselves could be given by an explicit formula. However, this formula is not easy to compute for a typical quadratic function. This paper shows that the weight recursions for the quadratic RS functions have an interesting special form which can be exploited to solve various problems about these functions, for example, deciding exactly which quadratic RS functions are balanced.
\end{abstract}


\section{Introduction}
\label{intro} 
If we define $V_n$ to be the vector space of dimension $n$ over the finite field $GF(2) = \set{0, 1}$, then an $n$ variable Boolean function $f(x_0, x_1, ..., x_{n-1}) = f(\vec{x})$ is a map from $V_n$ to $GF(2)$.  Every Boolean function $f(\vec{x})$ has a unique polynomial representation (usually called the {\em algebraic normal form} \cite[p. 6]{CBF}), and the {\em degree} of $f$ (notation $deg ~f$) is the degree of this polynomial.  A function of degree $\leq 1$ is \emph{affine}, and if the constant term is 0, then the function is \emph{linear}.  We let $B_n$ denote the set of all Boolean functions in $n$ variables, with addition and multiplication done $\bmod ~{2}.$ When addition $\bmod~{2}$ is clear from the context we use $+,$ but if we wish to emphasize the fact that addition is being done $\bmod~{2}$ we will use $\oplus.$ We also use $\oplus$ for the xor addition of two binary $m$-tuples.  We use $a \| b$ to denote the concatenation of two strings $a$ and $b.$

If we list the $2^n$ elements of $V_n$ as $v_0 = (0,\ldots,0), v_1 = (0,\ldots,0,1), \ldots$ in lexicographic order, then the $2^n$-vector $(f(v_0), f(v_1),\ldots,f(v_{2^n - 1}))$ is called the {\em truth table} of $f$.  The {\em weight} (also called Hamming weight) $wt(f)$ of $f$ is defined to be the number of 1's in the truth table for $f$.  In many cryptographic uses of Boolean functions, it is important that the truth table of each function $f$ has an equal number of 0's and 1's; in that case, we say that the function $f$ is {\em balanced}.  Another important kind of function in cryptography is the {\em bent} function (see \cite[Chapter 5]{CBF}), which is defined only if the number of variables is even.  A function $f$ in $n$ variables is bent if its distance (also called Hamming distance) from the set of all affine functions has its largest possible value $2^{n-1} - 2^{(n/2)-1}.$ We shall need the concept of the {\em nonlinearity} of a function $f$ (notation $N(f)$), which is defined to be the minimum distance from $f$ to any affine function. Thus $N(f)=0$ if and only if $f$ is an affine function and bent functions have the largest possible value for the nonlinearity.

Two Boolean functions $f(\vec{x})$ and $g(\vec{x})$ in $n$ variables are said to be {\it affine equivalent} if there exists an invertible matrix $A$ with entries in $GF(2)$ and ${\bf b} \in V_n$ such that $ f((\vec{x}))=g(A(\vec{x})\oplus {\bf b}).$
In general, determining whether or not two Boolean functions are affine equivalent is difficult, even in the simplest cases.  However, there is a simple test for the equivalence of quadratic functions, which is given in the following lemma.
\begin{lemma}
\label{quadeq}
Two quadratic functions $f$ and $g$ in $B_n$ are affine equivalent if and only if $wt(f) = wt(g)$ and $N(f) = N(g).$
\end{lemma}
\begin{proof}
This result has been well known for a long time, but there does not seem to be a proof in the literature before the one given in \cite[Lemma 2.3, p. 5068]{Cu11}.
\end{proof}

We shall also need the concept of {\em Walsh transforms}. The Walsh transform of a function $f$ in $n$ variables is the map $W_f: \mathbb{V}_n \rightarrow \mathbb{R}$ defined by 
$$ W_f({\bf w}) = \sum_{x \in V_n} (-1)^{f({\bf x}) +{\bf w} \cdot {\bf x}}, $$
where the values of $f$ are taken to be the real numbers $0$ and $1.$  The integers $ W_f({\bf w})$ are called {\em Walsh values}. We shall also need the well known formula
(see \cite[Th. 2.21, p. 17]{CBF})
\begin{equation}
\label{nonl}
N(f_n) = 2^{n-1}- \frac{1}{2} \underset{{\bf u}\in V_n} \max |W_f({\bf u})|.
\end{equation}

We define a cyclic permutation $\rho$ on $n$ variables by 
$\rho(x_0, x_1, \cdots, x_{n-1})=(x_1, x_2, \cdots, x_{n-1}, x_0)$. Then a Boolean function $f(x)$ in $n$ variables, where $x=(x_0, x_1, \cdots, x_{n-1})$, is  \emph{rotation symmetric} (RS for brevity) if $f(x)=f(\rho(x))$ for all $x \in V_n$.  A Boolean function is \emph{monomial rotation symmetric} (MRS for brevity) if it is rotation symmetric and generated by a single monomial. We use the notation 
$(0, a_1, \ldots, a_{d-1})_n$ for the monomial rotation symmetric function $f(x_0, x_1, \ldots, x_{n-1})$ of degree $d$ in $n$ variables generated by the monomial $x_0x_{a_1} \cdots x_{a_{d-1}}.$
 In \cite{PQ} Piepryzyk and Qu showed that rotation symmetric Boolean functions are useful  in cryptography for designing fast hash functions.  Since then, further applications of these functions in cryptography and coding theory have been found (many references for this are in \cite[Chapter 6]{CBF}), so much attention has been given to rotation symmetric Boolean functions. 
 
 \Cref{quadeq}, combined with the work in \cite{Kim09}, enables us to give a complete description of all of the affine equivalence classes for the quadratic MRS functions. This is explained in \Cref{affeq} below.

\subsection*{Acknowledgements}

A.C. was partially supported by NSF grant DMS-1801011. 
 
\section{Preliminaries}\label{se.prel}

We will need to recall a correspondence between quadratic RS functions on $GF(2)^n$ and functions of the form (here $Tr_n$ is the usual absolute trace on $GF(2^n)$ over $GF(2)$)
\begin{equation*}
  GF(2^n)\ni x \mapsto \sum_{i=1}^n a_i\mathrm{Tr_n}(x^{2^i+1})\in GF(2)
\end{equation*}
that have seen considerable interest (see e.g. \cite{Meidl17} and the many references therein). This correspondence was introduced in \cite[Definition 4.1]{cgl-sec}: given an RS quadratic function $Q$ in $n$ Boolean variables $x_1$ up to $x_n$, the corresponding map $Q'$ on $GF(2^n)$ is defined simply by
\begin{equation}\label{eq:13}
  Q'(x) = Q(x,x^2,\cdots,x^{2^{n-1}}). 
\end{equation}
The transformation $Q\mapsto Q'$ is less well behaved than one might hope: it does not, in general, preserve the nonlinearity or the weight (pathological behavior is noted throughout \cite[subsections 5.2 and 5.3]{cgl-sec}). Nevertheless, by \cite[Theorem 5.1]{cgl-sec} nonlinearity {\it is} preserved in the quadratic case we are concerned with here.

Since we are interested in whether or not quadratic functions are balanced, that is the property we will have to argue is preserved by $Q\mapsto Q'$. We have not been able to find this in the literature, and hence include a proof.

\begin{theorem}\label{th.qq'}
  Let $Q$ be a rotation symmetric quadratic function on $GF(2)^n$ and $Q'$ its corresponding trace representation defined by \Cref{eq:13}.

  Then, the Walsh transforms $W_{Q}({\bf 0})$ and $W_{Q'}({\bf 0})$ have the same absolute value:
  \begin{equation}\label{eq:15}
    W_Q({\bf 0})^2 = W_{Q'}({\bf 0})^2.
  \end{equation}
  In particular,
  \begin{equation}
    \label{eq:17}
    Q \text{ is balanced if and only if } Q' \text{ is}.
  \end{equation}  
\end{theorem}
\begin{proof}
  Since in general, for an $n$-variable Boolean function $f$ we have
  \begin{equation}\label{eq:14}
    W_f({\bf 0}) = 2^n-2wt(f)
  \end{equation}
  (from the definitions; see \cite[Lemma 2.10]{CBF}), $f$ is balanced if and only if its Walsh transform at ${\bf 0}$ vanishes. Hence, \Cref{eq:17} does indeed follow from \Cref{eq:15}.

  Conversely, suppose we show that $Q$ and $Q'$ are simultaneously (un)balanced. Then on the one hand this says that their Walsh transforms at ${\bf 0}$ vanish (or not) simultaneously. On the other hand, if one of the functions is not balanced then it follows from \Cref{le.parity-vectors} and \Cref{eq:14} that
  \begin{equation*}
    |W_f({\bf 0})|=2^n-2N(f).
  \end{equation*}
  Since $Q\mapsto Q'$ preserves the nonlinearity by \cite[Theorem 5.1]{cgl-sec}, \Cref{eq:15} follows from \Cref{eq:17}. In conclusion, it will be enough to prove the latter claim; the rest of the proof is devoted to this.

  Write
  \begin{equation*}
    Q(x) = \sum_{i=1}^{\lfloor\frac{n-1}2\rfloor}\sum_jc_i x_jx_{j+i},\ c_i\in GF(2)
  \end{equation*}
  where $j$ indices are considered modulo $n$. Following the proof of \cite[Lemma 1]{Gao12}, consider the matrix $(c_{i,j})_{i,j\mod n}$ defined by $c_{i,j}=c_{j-i}$ where the latter is defined and extended via $c_{n-i}=c_i$ where it initially was not (and $c_0=0$).

  Set also $C=B+B^t$, where $B$ is the (strictly) upper triangular part of $C$. We can then write $Q(x)=x^tBx$ (if we regard $x$ as a column vector), and one way to phrase \cite[Theorem 8.23]{car-bool} is to say that $Q$ is unbalanced precisely when the restriction of the quadratic form $x^tBx$ to the kernel of $C$ on the $n$-dimensional $GF(2)$-vector space $V$ is identically zero.

  Since $C$ is the matrix of the bilinear form associated to the quadratic form $Q(x)=x^tBx$, the restriction of the latter to the kernel of $C$ will be linear:
  \begin{align*}
    Q(x+y)&=Q(x)+Q(y),\quad \forall x,y\in \ker C \numberthis \label{eq:add} \\ 
    Q(cx) &= c^2Q(x)=cQ(x),\quad \forall x\in \ker C,\ c\in GF(2)
  \end{align*}
  (see also the remarks preceding \cite[Theorem 8.23]{car-bool}). 

  We now extend scalars to the $GF(2^n)$-vector space $W=V\otimes_{GF(2)}GF(2^n)$, again considering the quadratic form $Q$ and its associated bilinear form with matrix $C$ thereon; we sometimes write $C|_W$ to clarify that $C$ is regarded as the scalar-extended operator $C\otimes \mathrm{id}$ acting on $W=V\otimes_{GF(2)}GF(2^n)$ rather than $V$.

  On $\ker C|_W$ the restriction of the quadratic form $Q$ is {\it Frobenius-semi-linear}, in the sense that $Q$ satisfies the conditions in \Cref{eq:add} except that $x,y$ range over $W$, $c$ ranges over $GF(2^n)$, and the very last equality no longer holds.

  Note that since $C$ has entries in $GF(2)$ its kernel on $W$ is nothing but
  \begin{equation*}
    \ker C|_V\otimes_{GF(2)} GF(2^n). 
  \end{equation*}

  It follows from this and the noted semi-linearity that $Q$ is trivial on $C|_V$ if and only if it is trivial on $C|_W$. To summarize: $Q$ is balanced if and only if
  \begin{equation}\label{eq:18}
    Q(x)=x^tBx\text{ vanishes identically on }\ker C|_W. 
  \end{equation}

  Similarly, $Q'$ is unbalanced if and only if $Q'$ vanishes on the kernel of the additive (or {\it linearized}) polynomial
  \begin{equation}\label{eq:16}
    GF(2^n)\ni x\mapsto \sum_{i\le \lfloor\frac{n-1}2\rfloor} c_i (x^{2^{n-i}}+x^{2^i})\in GF(2^n).
  \end{equation}
  We now consider, following \cite{wl-lin}, the map $\iota$ defined by
  \begin{equation*}
    GF(2^n)\ni y\mapsto (y,y^2,\cdots,y^{2^{n-1}})^t\in W. 
  \end{equation*}
  Note that
  \begin{itemize}
  \item By the very definition of $Q'$ in \Cref{eq:13} we have
    \begin{equation*}
      Q'(x)=\iota(x)^tB\iota(x).
    \end{equation*}
  \item By \cite[equation (9)]{wl-lin} $\iota$ intertwines the polynomial \Cref{eq:16} and multiplication by $C$ on $W$, and hence also their kernels.     
  \end{itemize}

  In conclusion, the requirement that $Q'$ be trivial on the kernel of \Cref{eq:16} now simply means that the quadratic form $Q(x)=x^tBx$ is trivial on
  \begin{equation*}
    \ker C|_W\cap \iota(GF(2^n)). 
  \end{equation*}
  But according to \cite[Proposition 4.6]{wl-lin} this intersection is a $GF(2)$-structure on $\ker C|_W$, in the sense that a basis for it over $GF(2)$ is a basis for $\ker C|_W$ over $GF(2^n)$. In conclusion this latest vanishing condition is equivalent to \Cref{eq:18}, finishing the proof.
\end{proof}

\section{Weight recursions for MRS quadratics}
\label{secwtrec}
The paper \cite{C18} explains an algorithm for finding a linear recursion with integer coefficients for the Hamming weights 
$wt(f_n), n$ large enough,  where $f_n = f_n(x_0, x_1, \ldots, x_{n-1})$ is the algebraic normal form (polynomial form) of any RS Boolean function with degree $d$ in $n$ variables.  A Mathematica program which executes the algorithm is given in \cite{C17}. What the program actually does is use the algebraic normal form to compute a square matrix, called the {\it rules matrix}, whose minimal polynomial is a polynomial
$x^kp(x), k \geq 0,$ with integer coefficients such that $p(x)$ is the defining polynomial for the recursion satisfied by the weights. We say that $p(x)$ is the {\it recursion polynomial} for the weights $wt(f_n).$

In this section we prove some precise results about the recursion polynomials for quadratic MRS Boolean functions.  
Any such function must have its algebraic normal form equal to one of the functions
\begin{equation}\label{eq:3}
(0,t)_m = x_0 x_t+x_1 x_{t+1}+ \ldots + x_{m-1} x_{t-1}
\end{equation}
for some $t \geq 1$ or, if $m=2k$ is even and $t=k,$ equal to one of the functions
\begin{equation}
\label{short}
(0,k)_{2k} = x_0 x_{k}+x_1 x_{k+1}+ \ldots + x_k x_{2k}
\end{equation}
for some $k \geq 1.$  The functions of form \eqref{short} only have half as many monomials as the other quadratic MRS functions, and are called {\it short} functions. This terminology goes back at least to \cite[p. 5098]{Cu11}.

First we find an explicit form for the rules matrix for the quadratic functions $(0,t)_m.$ We need the cyclic permutation $\mu$ which acts on 
vectors $(b_1, b_2, \ldots, b_k)$ of any length $k$ by putting the last entry to the front, for example $\mu^2((0,1,0,1,0,0))=(0,0,0,1,0,1).$ Note this is the opposite direction of the permutation $\rho$ defined in the Introduction. 
\begin{theorem}
\label{th1}
We use $0_j$ to stand for a string of $j$ consecutive entries equal to $0$ and similarly for $1_j.$ The rules matrix for $(0,t)_m, t \geq 1,$ is a $2^{t}+1$ row square matrix $R'(t)$ which is obtained by adding a final column $(0_{2^{t}},2)^T$ and a final row 
$(0_{2^{t-1}},1_{2^{t-1}},2)$  to the $2^{t}$ row square matrix $R(t)$ whose rows are the $2^{t-1}$ pairs
$$\mu^i((1,0_{2^{t-1}-1},1,0_{2^{t-1}-1})),~\mu^i((1,0_{2^{t-1}-1},-1,0_{2^{t-1}-1}))$$
for $i=0, 1, \ldots, 2^{t-1}-1$ taken in order.  
\end{theorem}
\begin{proof}
  It is straightforward but lengthy to determine the entries in matrix $R'(t)$ by stepping through the Mathematica code which is given in \cite[Section 4]{C17}.  Of course any particular matrix $R'(t)$ or $R(t)$ for $t$ not too large (such as $R(4)$ in the example below) can be obtained by simply running the program.
\end{proof}
\begin{example}
\label{R4}
For $t=3,$ the set of $8$ rows of the matrix $R(3)$ in order is
\begin{multline*}
 \{(1, 0, 0, 0, 1, 0, 0, 0), (1, 0, 0, 0, -1, 0, 0, 0), (0, 1, 0, 0, 0, 1, 0, 0), (0, 1, 0, 0, 0, -1, 0, 0), \\ 
\shoveleft{(0, 0, 1, 0, 0, 0, 1, 0), (0,  0, 1, 0, 0, 0, -1, 0), (0, 0, 0, 1, 0, 0, 0, 1), (0, 0, 0, 1, 0, 0, 0, -1)\} }
\end{multline*}
The minimal polynomial of the $9$ row rules matrix is $x^7-2 x^6-8 x+16 = (x-2)(x^6-8)= (x-2)(x^2-2)(x^4+2x^2+4).$ Note it is obvious from the last row and column of $R'(t),$ as given in the Theorem, that $x-2$ will always be a factor of the minimal polynomial for $R'(t).$ 

This example shows that, if we let $wt((0,3)_n) = w(n),$ then the recursion for these weights has order $7$ and is given by 
$$w(n)=2w(n-1)+8w(n-6)-16w(n-7)$$
for $n \geq 7.$  
\end{example}
\begin{remark}
\label{rmk1}
The sequence of weights $w(n)$ starting with $n=5$ begins with $$16,28,64,112,256,480,1024,1792,4096,8064,16384$$ but the recursion is only correct if we start with $n=7.$ We discuss the reason for this later in the paper (see  \Cref{qbc} below).
It has to do with the fact that function $(0,3)_6$ is short and bent.
\end{remark}

We see from \Cref{R4} that the recursion polynomial for $(0,3)_m$ is {\it degenerate} (that is, it has at least two roots whose ratio is a root of unity--see \cite[p. 5]{recur}). The next theorem shows that the recursion polynomials for all the functions $(0,t)_m$ have a similar form.

\begin{theorem}
\label{th2}
The recursion polynomial for the weights of $(0,t)_m, t \geq 1,$ is
\begin{equation*}
  x^{2t+1} -2x^{2t} -2^{t}x  + 2^{t+1} = (x-2)(x^{2t}-2^{t}).
\end{equation*}
The recursion is valid for the sequence of weights $wt((0,t)_m)$ at least for all $m \geq 2t+1.$
\end{theorem}
\begin{proof}
By \Cref{th1} and the definition of the rules matrix $R'(t),$ it suffices to show that the minimal polynomial for $R(t)$ is 
$x^{2t}-2^{t}.$  The matrix $R(t)$ is sparse, so it is straightforward to compute its successive powers by partitioning each power into $4$ square submatrices, each with $2^{t-1}$ rows.  Define a sequence of square matrices $M(t)$ with $2^{t}$ rows such that $M(1)$ has first row $(1, 1)$ and second row $(1,-1).$  Then $M(t+1)$ is defined inductively as the square matrix partitioned into $4$ equal submatrices such that the top two submatrices are $M(t)$ and $M(t)$ and the bottom two submatrices are $M(t)$ and $-M(t).$ For example, the set of rows of $M(2)$ in order is
\begin{equation}
\label{rows3}
\{(1,1,1,1,),~(1,-1,1,-1),~(1,1,-1-1),~(1,-1,-1,1)\}.
\end{equation}
Routine computation using the products of the partitioned matrices $R(t)$ and $R(t)^k$ for $k=1, 2, \ldots$ shows that 
\begin{equation}
\label{Rt1}
R(t)^{t}=M(t),
\end{equation} 
that is the partition of $R(t)^{t}$ has  $M(t-1)$ and $M(t-1)$ as the top two matrices and $M(t-1)$ and $-M(t-1)$ as the bottom two matrices. For example, the set of rows of $R(3)^2$ is the set of rows of $M(3)$ in \eqref{rows3} above. Another computation shows
\begin{equation}
\label{Rt2}
M(t)^2=2^{t}I(2^{t}),
\end{equation}
where $I(j)$ is the identity matrix with $j$ rows.  Now \eqref{Rt1} and \eqref{Rt2} imply $R(t)^{2t}=2^{t}I(2^{t}),$ which proves that the minimal polynomial for $R(t)$ is $x^{2t}-2^{t}.$ 
\end{proof}
 Note that \Cref{th2} says that the minimal polynomial for the rules matrix $R'(t)$  is not divisible by $x,$ and so the minimal polynomial equals the recursion polynomial for $(0,t)_m.$  Note also that the matrices $M(t)$ in the proof of \Cref{th2} are well known special Hadamard matrices with $2^{t}$ rows (see for example \cite[pp. 20-21]{CBF}).
 
We can use the following two lemmas, taken from \cite{Kim09}, to get very precise results about the weight and nonlinearity of the quadratic MRS functions. 
\begin{lemma}
\label{Dick}
Suppose $f$ is a Boolean function in $n$ variables with degree $2$.  If  $f$  is balanced, then $f$ is affine equivalent to  $x_1x_2 + x_3x_4 + \ldots + x_{2d - 1}x_{2d} + x_{2d + 1}$ for some
$d \leq \frac{n - 1}{2}$.  If $f$ is not balanced, then $f$ is affine equivalent to
$x_1x_2 + x_3x_4 + \ldots + x_{2d - 1}x_{2d} + b$ for some $d  \leq \frac{n}{2}$ and $b$ in $GF(2).$ 
If $wt(f) < 2^{n - 1}$, then  $b = 0$.  If $wt(f) > 2^{n - 1}$, then  $b = 1$.
\end{lemma}
\begin{proof}
This result goes back to L. E. Dickson in $1901$ and is quoted in \cite[Th. 4, p. 429]{Kim09}. A proof is given in  \cite[pp. 438-442]{Mc78}.
Note that when $d=n/2$ for even $n,$ then $f$ is bent.
\end{proof}

\Cref{Dick} justifies introducing the following notions.

\begin{definition}
  Let $f$ be a quadratic Boolean function in $n$ variables. An element $a\in V_n$ is {\it $f$-parity-reversing} (or simply `parity-reversing' when $f$ is understood) if $f(x+a)=f(x)+1$ for all $x\in V_n$.

  Similarly, $a\in V_n$ is {\it $f$-parity-preserving} (or just `parity-preserving') if $f(x+a)=f(x)$ for all $x\in V_n$.
\end{definition}
The $f$-parity-preserving vectors form a $GF(2)$-vector subspace $V^0=V^0(f)$ of $V_n$, while the parity-reversing vectors form an affine subspace $V^1=V^1(f)$ that is either empty or a coset of $V^0(f)$.

In the context of \Cref{Dick}, in the balanced case the space $V^0$ of $f$-parity-preserving vectors are those whose coordinates $x_i$, $1\le i\le 2d+1$ vanish and similarly for the unbalanced case and $x_i$, $1\le i\le 2d$. Furthermore, in the balanced case $V^1(f)$ is the space of vectors with vanishing $x_i$, $1\le i\le 2d$ and $x_{2d+1}=1$. These remarks prove the following

\begin{lemma}\label{le.parity-vectors}
  Let $f$ be a quadratic Boolean function in $n$ variables. The following conditions are equivalent:
  \begin{enumerate}[(1)]
  \item $f$ is balanced;
  \item $\dim V^0(f)$ and $n$ have opposite parities;
  \item $V^1(f)$ is non-empty.    
  \end{enumerate}
\end{lemma}

Given a quadratic function $f,$ we call the affine equivalent function which is given in \Cref{Dick} the {\em Dickson form} of $f.$ We call the unique integer $d$ which \Cref{Dick} associates with a given function $f$ in $n$ variables the {\em Dickson rank} for $f.$ Given a function $f$ of degree 2, after we find the Dickson form for $f$ (this amounts to finding the Dickson rank, which unfortunately is not trivial to compute in general), it is easy to compute $wt(f)$ and $N(f)$.  The result is

\begin{lemma}
\label{wtnonlin}
Suppose $g$ is a Boolean function in $n$ variables which has the form $\sum_{i=1}^d x_{2i - 1} x_{2i} +
\sum_{i=2d + 1}^n a_i x_i$ with $d \leq \frac{n}{2}$, so $d$ is the Dickson rank for $f.$  Then $N(g)= 2^{n - 1} - 2^{n - d -1}$.  If all of the $a_i$ are $0$, then $wt(g) = N(g)$; otherwise $wt(g) = 2^{n - 1}$, so $g$ is balanced.
\end{lemma}
\begin{proof}
Two different proofs appear in \cite[pp. 438-442]{Mc78} and \cite[Lemma 5, p. 429]{Kim09}.
\end{proof}
To state the theorem below, we need the cyclic permutation $\rho_t$ defined on $\{0,1, \ldots, n-1\}$ by $\rho(0)=t, \rho(1)=t+1, \ldots, \rho(n-1) = t-1$ for $n \geq 2t+1.$ The function $(0,t)_n$ corresponds to $\rho_t$ in the obvious way and we say (as in \cite[p. 431]{Kim09}) $\rho_t$ is the {\em permutation of the function $(0,t)_n$}.
\begin{theorem}
\label{MRSwtnonlin}
Assume that $n \geq 2t+1$ and the permutation $\rho(t)$ of $(0,t)_n$ has the disjoint cycle decomposition $\rho_t= \mu_1 \mu_2 \ldots \mu_k.$  Then the number of cycles is  $k(n)=k =$ gcd$(n,t)$ and all cycles have the same length $n/k.$ We also have the formulas
$$wt((0,t)_n) = N((0,t)_n) = 2^{n-1} - 2^{(n/2)+k-1} ~\text{if}~n/k~\text{is even}$$
and
$$wt((0,t)_n) = 2^{n-1}~\text{and}~N((0,t)_n) = 2^{n-1} - 2^{(n+k-2)/2} ~\text{if}~n/k~\text{is odd}.$$
For the case $n=2t,$ the function is short and bent and we have $$wt((0,t)_{2t})=N((0,t)_{2t})=2^{2t-3} - 2^{t-2}.$$
\end{theorem}
\begin{proof}
The first two equations were given in \cite[Th. 8, p. 431]{Kim09}; the proof uses \Cref{Dick} and \Cref{wtnonlin}. These equations do not apply to the short and bent functions $(0,t)_{2t},$ since then the permutation $\rho$ is not defined.
\end{proof}
\begin{corollary}
We can directly compute the weight and nonlinearity for any function $(0,t)_n$ with $n \geq 2t+1$ from the values of $n$ and $k(n) =$
 gcd$(n,t).$
\end{corollary}

In stating some of our results below, we need the notion of a {\em plateaued} Boolean function.  This definition was introduced in $2001$ (see \cite[pp. 78-79]{CBF} for some of the history). We say that a Boolean function function $f=f_n$ in $n$ variables is  {\em v-plateaued} if
 every Walsh value $W_f({\bf w})$ is either $0$ or $\pm 2^{(n+v)/2}.$ This is the terminology of \cite[p. 266]{Meidl17} and it is convenient to use it in this paper; a more common usage would be to say that $f_n$ is {\em plateaued of order r} with $r=n-v$ \cite[Definition 4.26, p. 78]{CBF}; since $n$ and $v$ have the same parity, this order is always even. For given $f_n$ we say that $v=v(n)$ is the {\em v-value} of $f_n.$ Note it follows from \Cref{nonl} and \Cref{wtnonlin} that
 \begin{equation}\label{n=2d+v}
 n = 2d + v.
 \end{equation}
 We see from \Cref{n=2d+v} that $f_n$ is bent if and only if $v(n)=0.$
 
 It is well known that every quadratic Boolean function $f_n$ in $n$ variables is $v$-plateaued with $0 \leq v \leq n.$  The value $v=0$ occurs only when $n$ is even and $f_n$ is bent (see \Cref{Dick}). For MRS quadratic functions, $f_n$ is bent if and only if $f_n$ is the short function $(0,t)_{2t}$ for $t \geq 1.$  Many of our results apply only if bent functions do not occur, which explains the frequency of the presence of the condition $n \geq 2t+1.$ More generally, we have the following well known result.
 \begin{lemma} \label{qbc}
 If $f_{2t}$ is a quadratic rotation symmetric bent function, then the algebraic normal form must contain the function $(0,t)_{2t}.$
 \end{lemma}
 \begin{proof}
 See \cite[Remark 1, p. 4910]{Gao12}.
 \end{proof}
 
 The next lemma is a very special case of results in \cite[Section 2]{Meidl17} and gives a formula for the $v$-values of the MRS quadratic functions. As usual, we use $deg$ to denote the degree of a function.
 
 \begin{lemma}
 \label{vval}
 Given the quadratic Boolean function $(0,t)_n$ with $n \geq 2t+1,$ define
 \begin{equation}
A_t(x) = x^{t} +x^{n-t}.
 \end{equation}
 Then the $v$-values for $(0,t)_n$ are given by 
 \begin{equation}
 \label{vn}
 v(n) = \mathrm{deg}~ \mathrm{gcd}(x^n - 1, A_t(x)),
 \end{equation}
 where the greatest common divisor is taken mod $2.$
 \end{lemma}
 \begin{proof}
 This result is well known. An early published account of the relevance of the kind of gcd computation seen in \eqref{vn} is in \cite[Section 3]{KGS06}.
  \end{proof}
 The condition $n \geq 2t+1$ is needed to exclude the short and bent functions $(0,t)_{2t}$ (which have $v=0$), for which \Cref{vval} does not hold.
 
 In later sections of this paper we shall need the generalization of \Cref{vval} given below; a proof is in \cite[Section 2]{Meidl17}.
 
 \begin{lemma}
 \label{vvals}
 Given the quadratic Boolean function
 \begin{equation*}
   Q=Q(a_1, a_2, \ldots, a_{[(n-1)/2]})=\sum_{i=1}^{[(n-1)/2]} a_i (0,i)_n
 \end{equation*}
 with each $a_i$ in $\{0,1\}$ define
 \begin{equation}
 \label{fncA}
A(x) = \sum_{i=1}^{[(n-1)/2]} a_i (x^{i} +x^{n-i}).
 \end{equation}
 Then the $v$-values for $Q$ are given by 
 \begin{equation}
 \label{vng}
 v(n) = \mathrm{deg} ~\mathrm{gcd}(x^n - 1, A(x)),
 \end{equation}
 where the greatest common divisor is taken mod $2.$
 \end{lemma}
 
  \begin{theorem}
 \label{listvn}
 For each $t \geq 1,$ the sequence of $v$-values $v(n)$ for $(0,t)_n$ with $n \geq 2t+1$ is given by 
 \begin{equation}
 \label{vper}
 v(n) = \mathrm{gcd}(2t,n),~ n = 2t+1,~ 2t+2, \ldots
 \end{equation}
 \end{theorem}
 \begin{proof}
 For Boolean functions we have from \eqref{vn} $$v(n)= \mathrm{deg}~\mathrm{gcd}(x^n +1,x^{n-2t}+1)$$ and now the theorem follows from the 
 elementary fact $\mathrm{gcd}(x^i+1, x^j+1) = x^{\mathrm{gcd}(i,j)} +1$ for any positive integers $i$ and $j.$
\end{proof}

\begin{corollary}
  The integers $n$ and $v(n)$ have the same parity.
\end{corollary}
\begin{proof}
This result follows from \Cref{vper} and also from \Cref{n=2d+v}.
\end{proof}

  \begin{example}
  \label{ex1}
  The sequence of values $v(n),~n \geq 13,$ for $(0,6)_n$ begins with $1$, $2$, $3$, $4$, $1$, $6$, $1$, $4$, $3$, $2$, $1$, $12$ and has period of length $12$.
  \end{example}

 Now we can show that for any quadratic MRS function $(0,t)_n$ with $n \geq 2t+1$ we can compute the weight and Dickson rank, and hence by \Cref{wtnonlin} and \Cref{MRSwtnonlin} also the nonlinearity, using only $t$ and $n.$ For any integer $m$ we define
 \begin{equation}
 \label{2ind}
 \nu(m) = ~\text{largest integer $c$ such that}~ 2^c~ \text{divides}~ m.
 \end{equation}
 \begin{theorem}
 \label{nuth}
 Given the function $f=(0,t)_n$ with $t \geq 1$ and $n \geq 2t+1,$ 
 $f$ is not balanced if and only if $n \equiv 0 \bmod~{2^{\nu(t)+1}}.$  If $k(n) =~\text{gcd}(n,t)$ and $v(n)=~\text{gcd}(n,2t),$  then $v(n)=k(n)$ if $f$ is balanced and $v(n) = 2k(n)$ if $f$ is not balanced. Hence we can find the weight, Dickson rank and nonlinearity of $f.$
 \end{theorem}
 \begin{proof}
The first sentence in the theorem is true since by \Cref{MRSwtnonlin} $f$ is balanced if and only if $n/k(n)$ is odd.
 For any $n \geq 2t+1$ we can find the periodic sequence of $v$-values $v(n)$ using \Cref{listvn}. Now the values of $k(n)$ in \Cref{nuth} follow from \eqref{nonl}, the definition of plateaued and \Cref{MRSwtnonlin}. Once we have $k(n),$ the weight, nonlinearity and Dickson rank for $f$ follow from the first sentence in the theorem, \Cref{wtnonlin} and \Cref{MRSwtnonlin}.
 \end{proof}
 
 It follows from \Cref{listvn} and \Cref{nuth} that for $n \geq 2t+1$ the weights for the functions $(0,t)_n$ satisfy a recursion of order $2t+1,$ as we already saw in \Cref{th2}, and we know the recursion polynomial from \Cref{th2}.  

 We expect from the general theory of linear recurrences (see \cite[pp. 1-5]{recur} for the basics) that if we have an integer sequence 
 $\{a_j:~j=1, 2, \ldots \}$ which satisfies a recursion of order $N$ and if the recursion is {\em nondegenerate} (that is, the recursion polynomial has no pair of distinct roots whose ratio is a root of unity--see \cite[p. 5]{recur}) then there is a formula $a_j= \sum_{i=1}^N  c_i\alpha_i^j$ for the integers in the sequence, where the complex numbers $\alpha_j, 1 \leq j \leq N$ are the roots of the recursion polynomial.  \Cref{th2} shows that the recursions for
 $wt((0,t)_n)$ are always degenerate, but computations show that such formulas are nevertheless always true.                   
\section{Affine equivalence for MRS quadratics}
\label{affeq}
\Cref{nuth} shows that we can easily compute the nonlinearity, Dickson rank and weight for any function $(0,t)_n$ with $n \geq 2t+1,$ and the weight equals the nonlinearity if and only if the function is not balanced.  Our next lemma will enable us to exactly specify the values of $n$ for which $(0,t)_n$ is balanced.
\begin{lemma}
 \label{aa1}
The period for the $v-$values of the quadratic MRS function $(0,t)_n,$ $n \geq 2t+1,$ has length $2t.$ There is a unique largest integer $2t,$ which we place in the first position in the period.  The next largest entry $t$ is also unique, and occurs in position $t.$  The entries in the period are symmetric around position $t,$ that is, we have $v(jt+r) = v(jt-r)$ for each $r = 1, 2, ...$ for which both sides of the equation are defined and for each $j$ satisfying $1 \leq j \leq t.$
\end{lemma}
\begin{proof}
We already know the period length from \Cref{listvn}.  We cannot start the period at $v(2t)=0$ (a short and bent function) but it is convenient to begin the period with its largest element $2t.$ All of the other assertions in the lemma follow from $v(n) = \mathrm{gcd}(2t,n),~ n = 2t+1,~ 2t, \ldots$
given in \eqref{vper}.
\end{proof}
 
 \Cref{ex1} shows the period for $t=7,$ with the largest element moved to the end of the period. 
 
Using \eqref{nonl} and the fact that every quadratic function is plateaued, we see that \Cref{aa1} gives the nonlinearity for any function $(0,t)_n$ with $n \geq 2t+1.$  By \Cref{MRSwtnonlin} we also have the weight of any function $(0,t)_n$ which is not balanced, and we know that $(0,t)_n$ is balanced for $n \geq 2t+1$ if and only if $n/ \text{gcd}(n, t)$ is odd.   This proves the following lemma (compare \Cref{nuth}).

\begin{lemma}
\label{balMRS}
The function $(0,t)_n$ with $n \geq 2t+1$ is always balanced if $n$ is odd and is balanced for even $n$ if and only if the exact power of $2$ which divides $n$ is 
$\leq 2^{\nu(t)}.$  Equivalently, $(0,t)_n$ with $n \geq 2t+1$ is always balanced except when $n \equiv 0 \bmod 2^{\nu(t)+1}.$
\end{lemma}

We know from \Cref{quadeq} that each equivalence class for the MRS functions $(0,t)_n$ is uniquely determined by the pair of values $wt(f_n)$ and  $N(f_n)$ which all functions in the class have in common.  Using \Cref{MRSwtnonlin} and \Cref{balMRS} we can count the affine equivalence classes for the  MRS quadratic functions in $n \geq 3$ variables. We define $\tau(n)$ to be the number of positive integer divisors of $n.$

\begin{theorem}
\label{MRSeq}
The number of affine equivalence classes for the quadratic MRS functions $(0,t)_n, n \geq 3,$ is $\tau(n)-1.$ 
\end{theorem}
\begin{proof}
By \Cref{quadeq} all we need to do is count the number of distinct pairs $(wt((0,t)_n), N((0,t)_n))$ for $1<t \leq n/2$ if $n$ is even or for
$1<t \leq (n+1)/2$ if $n$ is odd.  We define $k(n)=k=gcd(n, t)$ as in \Cref{MRSwtnonlin}. Note that for each even $n,$ the short and bent function 
$(0,t)_{2t}$ with $k(n) = n/2$ is always in an equivalence class by itself.

For $n$ odd, \Cref{MRSwtnonlin} says $(0,t)_n$ always has weight $2^{n-1},$ and each possible value of $k$ gives a different nonlinearity.  Since $k=n$ gives $t>n$ (impossible) there are exactly 
$\tau(n) - 1$ different values of $k;$ so the case $k$ odd of \Cref{MRSeq} is proved.

Thus we can assume $n$ is even. By \Cref{MRSwtnonlin} we again need only count the number of choices for $k.$  If $n = 2^{\nu(n)} \prod_{i=1}^q p_i^{\beta(i)},$ where the $p_i$ are the distinct odd primes dividing $n,$ then when $n/k$ is even we must have $0 \leq \nu(k) < \nu(n).$ Thus the number of choices for $k$ in this case is 
\begin{equation}
\label{tau1}
\nu(n) \prod_{i=1}^q (\beta_i(n) + 1).
\end{equation}
When $n/k$ is odd, we must have $\nu(t) = \nu(n),$ so the number of choices for $k$ in this case is (recall $k=n$ is impossible)
\begin{equation}
\label{tau2}
 \prod_{i=1}^q (\beta_i(n) + 1) - 1.
\end{equation}
Adding \eqref{tau1} and \eqref{tau2}, we again obtain $\tau(n) - 1$ for the number of affine equivalence classes. 
\end{proof}

The authors of \cite{Kim09} deserve the credit for \Cref{MRSeq}, but they did not state it, perhaps because \Cref{quadeq} was not published when they did their work. 

\section{General RS quadratics}
\label{genquad}

We begin by considering sums $(0,s)_n+(0,t)_n, 1 \leq s < t,$ of two MRS quadratics.  If $t$ is not too large, the algorithm in \cite{C18} can be applied to find the recursion for the weights, but unlike the case for a single MRS function (see \Cref{th2}) there does not seem to be a simple formula for the order of the recursion.  For example, if we define
\begin{equation*}
  g_{n,i} = (0,1)_n + (0,i)_n, 2 \leq i \leq 5,
\end{equation*}
then the weight recursion orders for these functions are $5, 7, 17, 21,$ respectively.  The recursions apply only for $n \geq 2i+1.$ This happens because the presence of the short and bent functions $(0,i)_{2i}$ in the functions $g_{2i,i}$ means $wt(g_{2i,i})$ does not match the weight which the recursions would give. 
By using \Cref{vvals} we compute the $v$-values for the functions $g_{n,i}$ in the next example.
\begin{example}
Let $V(i)$ denote the list of the $v$-values $v(n,i)$ for the functions $g_{n,i}, 2 \leq i \leq 5,$ beginning with $n=2i+1$\\
$V(2)$ begins with $1, 4, 1, 2, 3, 2$  and has period length $6$\\ 
$V(3)$ begins with $1,6,1,2,1,4,1,2$ and has period length $8$\\
$V(4)$ begins with $3,6,1,4,1,~2,7,2,1,4,~1,6,3,2,1,~4,5,2,3,2,$\\
$1,8,1,2,3,~2,5,4,1,2$ and has period length $30$\\
$V(5)$ begins with $1, 8, 1, 2, 3, 6, ~1, 6, 1, 4, 3, 2, ~1, 10, 1, 2, 3, 4,~ 1, 6, 1, 6, 3, 2$\\
and has period length $24$
\end{example}

Computation shows that none of the functions $g_{n,3}$ and $g_{n,5}$ are balanced; also, the functions $g_{n,2}$ and $g_{n,4}$ are balanced only if $n \equiv 2 \bmod{4}.$ Since there does not seem to be a simple formula like \eqref{vper} for the $v$-values of these functions, it seems difficult to analyze the lists of $v$-values and to find a nice way to characterize the values of $n$ for which the functions are balanced (analog of \Cref{balMRS}). However, using new ideas we will make significant progress on this question in \Cref{cj.bal} below.

The next theorem gives some properties of the $v$-values. In the proof of this theorem we will work with Laurent rather than plain polynomials. The ring $L:=GF(2)[x^{\pm 1}]$ of Laurent polynomials over the field with two elements is still a Euclidean domain: this can be seen by extending the degree function from $P:=GF(2)[x]$ to $L$ by defining
\begin{align*}
  \deg(p(x))&= \max(\text{degree of a monomial of }p) \\
            &- \min(\text{degree of a monomial of }p)
\end{align*}
for every Laurent polynomial $p\in L$. 

We write $Q$ for a fixed polynomial $Q(a_i)$ as in \Cref{vvals} and
\begin{equation}\label{eq:20}
  A_n(x)=\sum_{i=1}^{J:=J(Q)} a_i(x^{i}+x^{n-i}),
\end{equation}
setting
\begin{equation}\label{eq:21}
  A(x):=A_0(x)=\sum_{i=1}^{J:=J(Q)} a_i(x^{i}+x^{-i})
\end{equation}
(a Laurent polynomial). Having fixed these conventions, we prove (note we take $n \geq 2J+1$ in order to avoid the short function $(0, J)_{2J}$)
\begin{theorem}\label{th.prd}
 Given any $Q(a_1, a_2, \ldots, a_{[(n-1)/2]}),$ the period for the list of $v$-values beginning with $n=2J(Q)+1$ has a unique largest integer $2J= 2 \max_{a_i \neq 0} i.$ The entries in the period are symmetric around this largest integer in the sense defined in \Cref{aa1}.  
\end{theorem}
\begin{proof}
  As previously, we write
  \begin{equation*}
    v(n) = \deg \mathrm{gcd} (x^n-1, A_n(x)).  
  \end{equation*}
  This is equal to
  \begin{equation}\label{eq:1}
    \deg \mathrm{gcd} (x^n-1,\ A(x)) = \deg\mathrm{gcd}\left(x^n-1,\ \sum_1^J a_i(x^{i}+x^{-i})\right),
  \end{equation}
since in every splitting field of $x^n-1$ over $GF(2)$ $x^n$ is identically $1$ (and hence we can eliminate the $n$ from $A_n$). 

Clearly, the maximal value of \Cref{eq:1} is $2J$: this is the degree (in the Laurent polynomial sense, as discussed above) of the right hand argument of the rightmost $\mathrm{gcd}$, and $n$ can be chosen so that
\begin{equation}\label{eq:2}
  A(x) | x^n-1\text{ in }L=GF(2)[x^{\pm 1}]:
\end{equation}
first choose an odd $m$ so that $x^m-1$ vanishes identically on a splitting field of $A(x)$, i.e. $m$ is the smallest positive integer of the form $2^k-1$ such that $z^m=1$ for all roots $z$ of $A(X)$ over the algebraic closure $\overline{GF(2)}$. Next, set $n=2^t m$ for the smallest $t$ with the property that $2^t$ dominates the multiplicity of every root of $A(x)$. Equation \Cref{eq:2} holds for $n=2^tm$ due to the identity
\begin{equation*}
  x^{2^tm}-1 = (x^m-1)^{2^t}
\end{equation*}
over $GF(2)$.

Note furthermore that $K=2^t(2^k-1)$ as described above has (by construction) the following properties:
\begin{equation}\label{Kprop1}
A(x) \text{ divides } x^n-1 \text{ if and only if } x^K-1 \text{ does, if and only if } K \text{ divides } n;
\end{equation}
and
\begin{equation}\label{Kprop2}
 K \text{ is a period for the sequence } \{v(n)\}, \text{ and in fact the smallest period. }
\end{equation}

These properties jointly ensure the uniqueness of $v(n)=2J$ for $n$ ranging over a contiguous sequence of $K$ values (i.e. $v(n)$ ranging over a period). In order to conclude, we have to prove the symmetry claim in the theorem. That claim, however, is now virtually immediate:

The Laurent polynomial $A(x)$ is invariant under the automorphism $x\leftrightarrow x^{-1}$ of $L$. Now, for every $0<n<K$ we have
\begin{equation*}
  v(tK-n) = \deg\mathrm{gcd}(x^{tK-n}-1,A(x)) = \deg\mathrm{gcd}(x^{-n}-1,A(x))
\end{equation*}
because $A(x)|x^{tK}-1$. This, in turn, equals
\begin{equation*}
\deg\mathrm{gcd}(x^{n}-1,A(x))
\end{equation*}
by the noted symmetry of $A(x)$ and finally, this is $v(tK+n)$ (once more because $A(x)$ divides $x^K-1$ and hence also $x^{tK}-1$).
\end{proof}

From now on, we consider a function $Q$ as defined in \Cref{vvals}. 
The next theorem is a main result in this paper. It gives an explicit description of those $n$ for which a quadratic function $Q$ as defined in \Cref{vvals} is balanced.  In view of \Cref{wtnonlin}, this means we can always determine the weight and nonlinearity (hence by \Cref{quadeq} also the affine equivalence class) of any quadratic RS function by a straightforward calculation of the $v$-value from \Cref{vvals}. Then the Dickson rank $d$ is given immediately by \Cref{n=2d+v}.

\begin{theorem}\label{cj.bal}
  Given any $Q(a_1, a_2, \ldots, a_{[(n-1)/2]}),$ if the number of nonzero $a_i$ is odd, then all functions $Q$  are balanced except for those with $n \equiv 0 \bmod 2^{c(Q)}$ for some integer $c(Q).$ If the number of nonzero $a_i$ is even, then either all functions $Q$  are not balanced or all functions $Q$ are not balanced except for those with $n \equiv 2^{d(Q)-1} \bmod 2^{d(Q)}$ for some integer $d(Q)$.
\end{theorem}

\Cref{MRSeq} is the first step towards solving the interesting problem of determining an exact count for the number of affine equivalence classes for quadratic RS functions in $n$ variables. The corresponding problem for cubic MRS functions was considered in \cite{Cu11}, but only for the special case of affine equivalence under permutations which preserve rotation symmetry. Later \cite{IJCM} this result was extended to affine equivalence under all permutations. The problem for general affine equivalence was out of reach then, and remains so. However, there is hope that the easier quadratic RS case can be attacked for general affine equivalence, and not just for MRS quadratics, because of the very simple necessary and sufficient condition for affine equivalence in \Cref{quadeq}.  We see from that lemma and the fact that all quadratic functions are $v$-plateaued for some $v$ (where $0 \leq v \leq n-2$ and $v$ and $n$ have the same parity) that all possible weights for a quadratic RS function in $n$ variables are $2^{n-1}$ (balanced function) and $2^{n-1} \pm 2^j,$ where $(n/2)-1 \leq j \leq n-2.$ Thus the number of possibilities for the pair $(weight, nonlinearity)$ is severely restricted. Computation for small $n$ shows that the smaller weights, corresponding to values of $v$ near $n,$ never seem to occur, which would further restrict the possibilities.  Thus the following question  can be raised.
\begin{question}\label{cj.B}
Is it possible that every  quadratic RS function in $n$ variables is affine equivalent to a function of form $Q(a_1, a_2, \ldots, a_{[(n-1)/2]})$ where the number of nonzero $a_i$ is $\leq B$ for some fixed integer $B?$
\end{question}
We have no example where even the very strong statement with $B=3$ is disproved.   

We will now begin to address \Cref{cj.bal}. Recall that the function $(0,t)_m$ in \Cref{eq:3} can be recast as
\begin{equation*}
  Q_m(x)=\mathrm{Tr}_m \left(x^{2^{t}+1}\right),\ x\in GF({2^m}),
\end{equation*}
where $\mathrm{Tr}_m:GF({2^m})\to \mathbb{F}_2$ is the degree-$m$ trace; see \cite{Meidl17}.  We shall often omit the subscript $m$ in $Q_m$ when it is clear from the context. It will often be convenient to use this trace form for the quadratic functions.   We say $Tr_m(x^{2^{t+1}})$ is balanced if and only if the truth table has $2^{m-1}$ $1$'s. This definition makes sense for $m \geq 1$ whereas $(0,t)_m$ is defined only for $m \geq t.$  We shall extend the definition of balanced by taking the trace definition for all $m \geq 1.$ By \Cref{th.qq'} the two definitions agree for $m \geq t.$

\begin{proposition}\label{pr.bal-odd}
  If $n$ is odd, a sum of an odd number of functions $(0,t)_n$ is balanced.
\end{proposition}
\begin{proof}
  Let
  \begin{equation}\label{eq:4}
    Q_n(x)=\mathrm{Tr}_n\left(\sum_i a_i x^{2^i+1}\right),\ a_i\in GF(2)
  \end{equation}
  be a function as in the statement, with an odd number of non-zero $a_i\in GF(2)$. We claim that under the hypotheses we have $Q(x+1)=Q(x)+1$, which would clearly entail the desired conclusion that the preimages
  \begin{equation*}
    Q^{-1}(0)\text{ and }Q^{-1}(1)\subset GF({2^n})
  \end{equation*}
  have the same cardinality.

  Given that we are assuming the sum in \Cref{eq:4} has an odd number of non-zero terms, it will suffice to show that
  \begin{equation}\label{eq:5}
    \mathrm{Tr}_n\left((x+1)^{2^i+1}\right)  = \mathrm{Tr}_n\left(x^{2^i+1}\right) +1. 
  \end{equation}
  To see this, note first that
  \begin{equation*}
    (x+1)^{2^i+1} = (x^{2^i}+1)(x+1) = (x^{2^i+1}+1) + (x^{2^i}+x). 
  \end{equation*}
  The second term on the right hand side has zero trace: the traces of $x$ and $x^{2^i}$ coincide, since the latter is the image of the former through an iteration of the Frobenius automorphism $x\mapsto x^2$ of the field $GF({2^n})$.

  It follows that 
  \begin{equation*}
    \mathrm{Tr}_n\left((x+1)^{2^i+1}\right)  = \mathrm{Tr}_n\left(x^{2^i+1} + 1\right),
  \end{equation*}
  which is nothing but \Cref{eq:5} once we observe that $\mathrm{Tr}_n(1)=n=1$ because $n$ is assumed odd.
\end{proof}

The technique employed in the proof of \Cref{pr.bal-odd} extends to provide a sufficient condition for balancing under more general circumstances. To state the result, we use the function $\nu(n)$ defined in \eqref{2ind} (i.e. the $2$-adic valuation of $n$).

\begin{theorem}\label{pr.bal-suf}
  Let $Q(x)$ be a function as in \Cref{eq:4} with an odd number of terms. If
  \begin{equation*}
    \nu(n)\le \min_{a_i\ne 0}\nu(i)
  \end{equation*}
  then $Q$ is balanced. 
\end{theorem}
\begin{proof}
  Write $n=2^{\nu}m$ with $m$ odd and $\nu=\nu(n)$. The trace $\mathrm{Tr}_n$ is then the composition of two intermediate traces:
  \begin{equation}\label{eq:9}
    \mathrm{Tr}_n = \mathrm{Tr}_{2^{\nu}}\circ \mathrm{Tr}_{GF({2^n})/GF({2^{2^{\nu})}}}. 
  \end{equation}
  Denote the rightmost trace by $\mathrm{Tr}_{mid}$ for brevity. If we show that
  \begin{equation*}
    x\mapsto Q_{mid}(x):=\mathrm{Tr}_{mid}\left(\sum_i a_i x^{2^i+1}\right)\in GF({2^{2^{\nu})}}
  \end{equation*}
  achieves every value in its codomain $GF({2^{2^{\nu})}}$ the same number of times (or in short, is balanced as a $GF({2^{2^{\nu})}}$-valued function) then we can conclude that $Q$ is balanced by simply composing further with $\mathrm{Tr}_{2^{\nu}}$, which has the same property (i.e. the preimages $\mathrm{Tr}_{2^\nu}^{-1}(0)$ and $\mathrm{Tr}_{2^\nu}^{-1}(1)$ are equinumerous).

  In turn, proving that $Q_{mid}$ is balanced will follow from the equation
  \begin{equation}\label{eq:6}
    Q_{mid}(\bullet+a) = Q_{mid}(\bullet)+a^2,\ \forall a\in GF({2^{2^\nu})}. 
  \end{equation}
  To see this, note first that our assumption on  $2$-adic valuations ensures that for all $x^{2^i}$ appearing in the expression of $Q_{mid}$ we have $\nu(i)\ge \nu$ and hence $x\mapsto x^{2^i}$ is an iterated application $x\mapsto F^{d_i}x$ of the Frobenius automorphism
  \begin{equation*}
    F:x\mapsto x^{2^{2^\nu}}
  \end{equation*}
  of $GF({2^{2^{\nu})}}$. In conclusion, for each term $x^{2^i+1}$ of $Q_{mid}$, we have
  \begin{equation*}
    (x+a)^{2^i+1} = F^{d_i}(x+a)\cdot (x+a) = (F^{d_i}x+a)(x+a). 
  \end{equation*}
  The two terms $F^{d_1}x\cdot a$ and $a\cdot x$ cancel out upon taking the trace $\mathrm{Tr}_{mid}$, so that leaves us with
  \begin{equation*}
    F^{d_i}x\cdot x + a^2 = x^{2^i+1}+a^2.
  \end{equation*}
  Applying $\mathrm{Tr}_{mid}$ to $a^2$ produces
  \begin{equation*}
    [GF({2^n}):GF({2^{2^\nu})}]a^2 = a^2
  \end{equation*}
  because the degree $[GF({2^n}):GF({2^{2^\nu})}]=m$ is odd, and finally the fact that we have an odd number of such terms $x^{2^i+1}$ proves \Cref{eq:6} and hence the theorem.
\end{proof}

In particular, when $Q$ has a single term, we recover the sufficiency condition for balancing obtained previously in \Cref{nuth}. 

We write $n=km$ for odd $m$ and decorate the function $Q$ in \Cref{eq:4} with an `$n$' subscript to emphasize that it involves an application of $\mathrm{Tr}_n$. This will allow us to talk about the analogues
\begin{equation*}
  Q_d(x) = \mathrm{Tr}_d\left(\sum_i a_i x^{2^i+1}\right) : GF({2^d})\to GF(2)
\end{equation*}
for every divisor $d|n$. Recall also the decomposition \Cref{eq:9} of $\mathrm{Tr}_n$; in the present setup we once more write it as
\begin{equation}\label{eq:11}
  \mathrm{Tr}_n = \mathrm{Tr}_k\circ \mathrm{Tr}_{mid},
\end{equation}
where the factor is the intermediate trace $GF({2^n})\to GF({2^k})$

The following observation, which builds on the proof of \Cref{pr.bal-suf}, will come in handy later.

\begin{lemma}\label{le.trnsl}
  Let $Q$ be a quadratic function as in \Cref{eq:4}. Then, for every $x\in GF({2^n})$ with
  \begin{equation*}
    GF({2^{k})}\ni \mathrm{Tr}_{mid}(x)=0
  \end{equation*}
  and $a\in GF({2^{k})}$ we have
  \begin{equation}\label{eq:8}
    Q_n(x+a) = Q_n(x) + Q_{k}(a). 
  \end{equation}
\end{lemma}
\begin{proof}
  It is enough to prove this for a single term $x\mapsto x^{2^i+1}$, $i\ge 1$ of $Q$. We have
  \begin{equation}\label{eq:7}
    (x+a)^{2^i+1} = (x^{2^i}+a^{2^i})(x+a) = x^{2^i+1} + (x^{2^i}a + xa^{2^i+1}) + a^{2^i+1}. 
  \end{equation}
  The three terms of the rightmost expression in \Cref{eq:7} are disposed of as follows.

  \begin{itemize}
  \item Applying $\mathrm{Tr}_n$ to the first term produces $Q_n(x)$ in \Cref{eq:8}.
  \item Applying
    \begin{equation*}
      \mathrm{Tr}_n = \mathrm{Tr}_{k}\circ \mathrm{Tr}_{mid}
    \end{equation*}
    to the third term produces
    \begin{equation*}
      Q_{k}(a) = \mathrm{Tr}_{k}\left(a^{2^i+1}\right)
    \end{equation*}
    because $\mathrm{Tr}_{mid}$ is the relative trace of an odd-degree field extension of $GF({2^{k})}$ and $a$ belongs to the latter field.
  \item Finally, $\mathrm{Tr}_n$ annihilates the second term $x^{2^i}a + xa^{2^i+1}$ on the right hand side of \Cref{eq:7} because $\mathrm{Tr}_{mid}$ does: the latter produces
    \begin{equation*}
      \mathrm{Tr}_{mid}(x)^{2^i}a + \mathrm{Tr}_{mid}(x)a^{2^i}
    \end{equation*}
    because (in characteristic two) traces commute with squaring, and we are assuming that $\mathrm{Tr}_{mid}(x)$ vanishes.
  \end{itemize}
  Jointly, these three remarks prove the desired conclusion.
\end{proof}

The usefulness of the lemma will become apparent in the course of the proof of the following result.

\begin{theorem}\label{pr.2-adic-suf}
  Let $Q(x)$ be a quadratic function as defined as in \Cref{eq:4} and consider a positive integer $n=k m$ for odd $m$. If $Q_{k}$ is balanced then so is $Q_n$. 
\end{theorem}
\begin{proof}
  Because the degree $m$ of the extension
  \begin{equation}\label{eq:10}
    GF({2^{k})}\subseteq GF({2^n})
  \end{equation}
  is odd, every element of the larger field $F_n$ can be written (uniquely) as $x+a$ where
  \begin{itemize}
  \item $x$ is annihilated by the relative trace $\mathrm{Tr}_{mid}$ of \Cref{eq:10};
  \item $a$ belongs to the smaller field $GF({2^{k})}$. 
  \end{itemize}
  But then, according to \Cref{le.trnsl} we have 
  \begin{equation*}
    Q_n(x+a) = Q_n(x) + Q_{k}(a). 
  \end{equation*}
  Since we are assuming that $Q_{k}$ is balanced, this implies that every coset of $GF({2^{k})}$ in $GF({2^n})$ contains equal numbers of elements in $Q_n^{-1}(0)$ and $Q_n^{-1}(1)$. Since $GF({2^n})$ is a disjoint union of such cosets, this proves the conclusion that $Q_n$ is balanced.
\end{proof}

We will now reverse the implication in \Cref{pr.2-adic-suf}:


\begin{theorem}\label{th.val-iff}
  Let $Q(x)$ be a quadratic function as defined as in \Cref{eq:4} and consider a positive integer $n=k m$ for odd $m$. Then, $Q_{k}$ is balanced if and only if $Q_n$ is.

  In particular, whether or not $Q_n$ is balanced depends only on the $2$-adic valuation of $n$. 
\end{theorem}
\begin{proof}
  The last statement follows from the rest. As for the first statement, one implication is covered by \Cref{pr.2-adic-suf}, so we focus on the converse.

  According to \Cref{le.parity-vectors} the function $Q_n$ is balanced if and only if there is a $Q$-parity-reversing element $a\in GF({2^n})$, i.e. one satisfying
  \begin{equation}\label{eq:12}
    Q_n(x+a) = Q_n(x)+1,\ \forall x\in GF({2^n});
  \end{equation}
  
The factorization \Cref{eq:11} implies that
  \begin{equation*}
    Q_n|_{GF({2^k})} = Q_k:
  \end{equation*}
  indeed, $\mathrm{Tr}_{mid}$ is the identity on $GF({2^k})$ because the degree $m=[GF({2^n}):GF({2^k})]$ is odd. In conclusion, the claimed equivalence will follow once we show that an element $a\in GF({2^n})$ satisfying \Cref{eq:12}, if it exists, can be chosen in $GF({2^k})$.

  To see this, let $a\in GF({2^n})$ be parity-reversing. Then, since $Q$ is idempotent (i.e. $Q(x^2)=Q(x)$) all elements
  \begin{equation*}
    a,\ a^{2^k},\ \cdots,\ a^{2^{k(m-1)}}
  \end{equation*}
  have the same property. Since there are $m$ of them, i.e. an odd number,
  \begin{equation*}
    \mathrm{Tr}_{mid}(a) = a+a^{2^k}+\cdots+a^{2^{k(m-1)}}\in GF({2^k})
  \end{equation*}
  is again parity-reversing. As noted, this concludes the proof.
\end{proof}

\Cref{th.val-iff} reduces the problem of whether or not $Q_n$ is balanced to the case when $n$ is a power of $2$. This allows us to supplement \Cref{pr.bal-suf} with a converse:

\begin{theorem}\label{pr.nlow}
  If $Q$ is as in \Cref{eq:4} and $n$ satisfies
  \begin{equation*}
    \nu(n)\le \min_{a_i\ne 0}\nu(i)
  \end{equation*}
  then $Q_n$ is balanced if and only if $Q$ has an odd number of terms.
\end{theorem}
\begin{proof}
  Write $n=2^\nu m$ for odd $m$, as in the proof of \Cref{pr.bal-suf}. According to \Cref{th.val-iff} $Q_n$ is balanced if and only if $Q_{2^\nu}$ is, so it is enough to assume that $n=2^\nu$. But then, for each summand $x^{2^i+1}$ of $Q$ the map $x\mapsto x^{2^i}$ is an iterated application of the Frobenius automorphism of $GF({n})$ and thus the identity as a function on $\mathbb{F}_n$.

  It follows that every term of $Q_n$ is $\mathrm{Tr}_n(x^2)$ and hence $Q_n$ is either $\mathrm{Tr}_n(x^2)$ (and balanced) when $Q$ has an odd number of terms or identically zero otherwise.
\end{proof}

In particular:

\begin{corollary}\label{cor.odd-n}
If $n$ is odd then $Q_n$ is balanced if and only if $Q$ has an odd number of summands.   
\end{corollary}
\begin{proof}
  This is a direct application of \Cref{pr.nlow}.
\end{proof}

We can now tackle another particular case of \Cref{cj.bal}.

\begin{proposition}\label{pr.all-odd}
  If $Q$ is as in \Cref{eq:4} and all $i$ appearing in the terms $x^{2^i+1}$ of $Q$ are odd then $Q_n$ is balanced if and only if
  \begin{itemize}
  \item $n$ is odd, and
  \item $Q$ has an odd number of terms.
  \end{itemize}
\end{proposition}
\begin{proof}
  The fact that for odd $n$ being balanced is equivalent to having an odd number of terms is \Cref{cor.odd-n}, so it is enough to prove that $Q_n$ as in the statement cannot be balanced for even $n$. According to \Cref{le.parity-vectors} this is equivalent to showing that the space $V^0(Q)\le GF({2^n})$ of $Q$-parity-preserving vectors is even-dimensional. We will argue that in fact $V^0(Q)$ is a vector space over $\mathbb{F}_4\subseteq GF({2^n})$, which will imply the desired conclusion. 

  Factor $\mathrm{Tr}_n$ as $\mathrm{Tr}_2\circ \mathrm{Tr}_{mid}$, where
  \begin{equation*}
    \mathrm{Tr}_{mid}:GF({2^n})\to \mathbb{F}_4
  \end{equation*}
  is the intermediate trace. Let $x\in GF({2^n})$ and $a\in \mathbb{F}_4$. For every term $R(x)=x^{2^i+1}$ of $Q$ we have 
  \begin{equation*}
    R(xa) = (xa)^{2^i+1} = x^{2^i+1}a^{2^i+1} = x^{2^i+1}a^3
  \end{equation*}
  because $i$ is assumed odd. Now, $a^3$ is either $0$ or $1$ depending on whether $a\in \mathbb{F}_4$ vanishes or not. It follows that $Q_n(xa)$ is either $0$ on $a=0$ or $Q_n(x)$ otherwise. This implies that $V^0(Q)$ is invariant under multiplication by $\mathbb{F}_4\subseteq GF({2^n})$, which is what we sought to prove.
\end{proof}

Fix a function $Q$ as in \Cref{eq:4}. We will now see that \Cref{th.val-iff} imposes strong restrictions on the set of positive integers $n$ for which $Q_n$ is balanced. First, recall that by \cite[Theorem 1]{C18} the weights $w(n)$ of $Q_n$ satisfy a linear recurrence with integer coefficients. Now, weights are of the form
\begin{equation*}
  w(n) = 2^{n-1}\pm 2^{\frac {n+v(n)}2},
\end{equation*}
and hence the linearly recurrent sequence $nw(n)=\frac{w(n)}{2^{n-1}}$ (`$nw$' for `normalized weight') is of the form
\begin{equation*}
  nw(n)=1\pm 2^{\frac{v(n)-n+1}2}, 
\end{equation*}
and being balanced is equivalent to $nw(n)=1$.

Since the celebrated theorem of Skolem-Mahler-Lech (e.g. \cite[Theorem 5.1]{csl}) ensures that the level sets of a linearly recursive sequence are (essentially) finite unions of arithmetic progressions, we have

\begin{proposition}\label{pr.sml}
  Given $Q$, there is a positive integer $N=N(Q)$ and a set $\mathcal{R}=\mathcal{R}(Q)$ of residues modulo $N$ such that, for sufficiently large $n$, $Q_n$ is balanced if and only if $n~(\mathrm{mod}\ N)\in \mathcal{R}$.
\end{proposition}

\Cref{th.val-iff} supplements this picture considerably: it tells us that the set $N(Q)\mathbb{N}+\mathcal{R}(Q)$ of positive integers giving residues in $\mathcal{R}$ modulo $N$ is (except perhaps for finitely many terms) invariant under multiplication and division by odd positive integers. This implies the following

\begin{theorem}\label{th.pows2}
  Given $Q$, there are finite sets $\mathcal{S}=\mathcal{S}(Q)$ and $\mathcal{T}=\mathcal{T}(Q)$ of positive integers such that $Q_n$ is balanced if and only if
  \begin{equation*}
    n\equiv 2^{d-1}~ \mathrm{mod}~ 2^d\text{ for some }d\in \mathcal{S}
  \end{equation*}
  or
  \begin{equation*}
    n\equiv 0~ \mathrm{mod}~ 2^d\text{ for some }d\in \mathcal{T}.
  \end{equation*}
\end{theorem}

\begin{remark}\label{re.t-sing}
  Note that the set $\mathcal{T}(Q)$ in \Cref{th.pows2}, when non-empty, might as well be a singleton. \Cref{cj.bal} implies in particular that it {\it is} always empty, which we prove in \Cref{th.no-t} below.
\end{remark}

\subsection{Balanced functions and linearized polynomials}\label{subse.lon}

Let $Q$ be a quadratic function of the form \Cref{eq:4}. 

Recall also the balancedness criterion used in the proof of \Cref{th.qq'} (and also the proof of \cite[Theorem 5.1]{cgl-sec}, adapted from \cite[Theorem 8.23, p. 312]{car-bool}): $Q$ is balanced on $GF(2^n)$ if and only if its restriction to the kernel of the additive polynomial
\begin{equation}\label{eq:19}
  F_Q(x):=\sum_i a_i\left(x^{2^{n-i}}+x^{2^i}\right). 
\end{equation}
(regarded as a $GF(2)$-endomorphism of $GF(2^n)$) vanishes identically.

Denoting by $F$ the Frobenius automorphism $x\mapsto x^{2}$ in characteristic two, \Cref{eq:19} is an application to $x$ of $A_n(F)$, where $A_n$ is the polynomial \Cref{eq:20} attached to $Q$. Equivalently, we can work with the {\it Laurent} polynomial $A$ defined in \Cref{eq:21} applied to $F$ (since the latter is invertible as an endomorphism of $GF(2^n)$).

On $GF(2^n)$ the Frobenius morphism $F$ is annihilated by $\psi_n(x):=x^n-1$ (i.e. $F^n=\mathrm{id}$). For that reason, the kernel of $A(F)$ will also coincide with the kernel of $\mathrm{gcd}(\psi_n,A_n)$ (the same polynomial appearing in \Cref{th.prd}). We will pass freely between plain and the Laurent polynomials $GF(2)[x^{\pm 1}]$ and hence work with $\mathrm{gcd}(\psi_n,A)$, etc. Note in particular that $\mathrm{gcd}(\psi_n,A_n)$ has degree
\begin{equation*}
v(n)\le \max\{2i\ |\ a_i\ne 0\}
\end{equation*}
and hence ranges over finitely many possibilities. 

This characterization of balancedness will allow us to eliminate one of the possibilities listed in \Cref{th.pows2} which would contradict \Cref{cj.bal} (see \Cref{re.t-sing}).

\begin{theorem}\label{th.no-t}
  For $Q$ as in \Cref{eq:4} there is some $\nu$ such that $Q_n$ is unbalanced as soon as the $2$-adic valuation of $n$ is $\ge \nu$.

  In particular, in \Cref{th.pows2} the set $\mathcal{T}$ is empty. 
\end{theorem}
\begin{proof}
  Indeed, choose $n$ so that the splitting field of the polynomial $F_Q$ defined in \Cref{eq:19} is contained in $GF(2^{\frac n2})$. Then, all elements $x\in GF(2^n)$ annihilated by the product of the finitely many polynomials (recall $\psi_n(x)=x^n-1$) 
  \begin{equation*}
    \mathrm{gcd}(\psi_n,A_n)(F)x
  \end{equation*}
  are contained in the subfield
  \begin{equation*}
    GF(2^{\frac n2})\subset GF(2^n)
  \end{equation*}
  (i.e. $GF(2^{\frac n2})$ is a splitting field for said product). The conclusion follows from the fact that the trace
  \begin{equation}\label{eq:22}
    \mathrm{Tr}_n:GF(2^n)\to GF(2)
  \end{equation}
  vanishes on $GF(2^{\frac n2})$.
\end{proof}

Consequently, we have the following improved version of \Cref{th.pows2}.

\begin{corollary}\label{cor.no-t}
  Given $Q$, there is a finite set $\mathcal{S}=\mathcal{S}(Q)$ of positive integers such that $Q_n$ is balanced if and only if
  \begin{equation*}
    n\equiv 2^{d-1}\ \mathrm{mod}\ 2^d\text{ for some }d\in \mathcal{S}
  \end{equation*}
\end{corollary}

With this in place, \Cref{cj.bal} says that $\mathcal{S}(Q)$ is a singleton if $Q$ has an even number of terms and an initial segment of $\mathbb{Z}_{\ge 0}$ otherwise.

Since we are concerned mostly with the question of whether $Q$ is balanced, \Cref{th.val-iff} allows us to assume that $n$ is a power of $2$: $n=2^{\nu}$; we do so throughout the present discussion, unless specified otherwise. We now record a number of additional remarks on the characterization of balancedness discussed here.

Under the assumption that $n=2^{\nu}$, $\psi_n(x)$ is simply $(x-1)^{2^{\nu}}$. In conclusion,
\begin{equation*}
  \mathrm{gcd}(\psi_n,A)=(x-1)^{v(n)},
\end{equation*}
where $v(n)$ is the plateau parameter that is the focus of \Cref{th.prd}. Note that in this particular case, where $n=2^{\nu}$,
\begin{equation*}
  v(n)=\max\{d\le n\ |\ (x-1)^d\text{ divides }A(x)\}.
\end{equation*}
In short, we are interested in whether or not the restriction of $Q$ to
\begin{equation*}
  x\in GF(2^n),\ (F-\mathrm{id})^{v(n)}x=0
\end{equation*}
(where the exponent on the right hand side denotes repeated composition) is identically zero.


Write $d_Q$ for the largest exponent such that $(x-1)^{d_Q}$ divides the Laurent polynomial $A(x)$ from \Cref{eq:21}. Regarding
\begin{equation*}
  \ker(F-\mathrm{id})^{d_Q}
\end{equation*}
as a subspace of a fixed algebraic closure $\overline{GF(2)}$, $Q$ is unbalanced if and only if it vanishes identically along 
\begin{equation*}
  \ker(F-\mathrm{id})^{d_Q}\cap GF(2^n). 
\end{equation*}
We also introduce the notation $\nu_Q$ for the number defined uniquely by
\begin{equation}\label{eq:28}
  2^{\nu_Q-1}< d_Q\le 2^{\nu_Q}. 
\end{equation}

The following result is a quantitative enhancement of \Cref{th.no-t}.

\begin{proposition}\label{pr.bet-pow2}
  Suppose $Q_n$ is balanced for some $n=2^{\nu}$. Then, $\nu \le \nu_Q$.  
\end{proposition}
\begin{proof}
  Suppose not. All $x\in \overline{GF(2)}$ annihilated by $(F-\mathrm{id})^{d_Q}$ are contained in $GF(2^{2^{\nu_Q}})$, and our assumption is that the latter field is contained strictly in
  \begin{equation*}
    GF(2^n) = GF(2^{2^{\nu}}).
  \end{equation*}
  It follows that the trace \Cref{eq:22} vanishes on $\ker(F-\mathrm{id})^{d_Q}$ and hence $Q$ is not balanced. This provides the requisite contradiction.
\end{proof}

We also record the following variant (and consequence) of \Cref{pr.bet-pow2}.

\begin{corollary}\label{pr.2period}
Let $N$ be the period of the sequence $v(n)$ of plateau parameters and suppose $Q_n$ is balanced for some $n=2^{\nu}$. Then, $\nu \le \nu(N)$.     
\end{corollary}
\begin{proof}
  It follows from \Cref{Kprop1} and \Cref{Kprop2} that the period $N$ is the smallest positive integer for which $x^N-1$ is divisible by the Laurent polynomial $A(x)$ in \Cref{eq:21}. This means in particular that
  \begin{equation*}
    (x-1)^{d_Q} \text{ divides } A(x) \text{ divides } x^N-1;
  \end{equation*}
  therefore $N$ is divisible by the smallest power of $2$ that dominates $d_Q$ (namely $2^{\nu_Q}$, by \Cref{eq:28}). We thus have $\nu_Q\le \nu(N)$, and the conclusion follows from \Cref{pr.bet-pow2}.
\end{proof}

A consequence of \Cref{pr.bet-pow2}:

\begin{proposition}\label{pr.cutoff}
  Suppose $Q_n$ is balanced for some $n=2^{\nu}$. Then, for any smaller power of two $n'=2^{\nu'}$, $\nu'<\nu$, $Q_{n'}$ is unbalanced if and only if it is identically $0$.
\end{proposition}
\begin{proof}
  By \Cref{pr.bet-pow2} we have $\nu \le \nu_Q$ and hence $\nu-1\le \nu_Q-1$. But note that
  \begin{equation}\label{eq:23}
    \ker(F-\mathrm{id})^{d_Q}\subset \overline{GF(2)}
  \end{equation}
  {\it contains} $GF(2^{2^{\nu_Q-1}})$ (because by definition $d_Q> 2^{\nu_Q-1}$) and thus also $GF(2^{n'})\subseteq GF(2^{2^{\nu-1}})$. 

  Being unbalanced over $GF(2^{n'})$ is equivalent to vanishing on \Cref{eq:23} and hence identically by the previous paragraph.
\end{proof}

In fact, essentially the same argument proves

\begin{proposition}\label{pr.0orall}
  If $2^\nu\le d_Q$ then $Q_n$ is unbalanced if and only if it is identically zero.
\end{proposition}

Since \Cref{pr.cutoff} renders meaningful the question of whether or not $Q_n$ vanishes identically, we examine that problem in more detail. It will occasionally be convenient to work with RS functions of the form
\begin{equation}\label{eq:24}
  \sum_i \sum_{j~\mathrm{mod}~n} a_i x_j x_{j+i} = \sum_{i \neq (n/2)} a_i(0,i)_n
\end{equation}
defined on $GF(2)^n$. The omitted value of $i$ for $n$ even on the right-hand side corresponds to the short function $(0, \frac{n}{2})_n$ (see \Cref{short}) which would appear twice in \Cref{eq:24}. This value of $i$ corresponds to the "additional term" in \Cref{pr.when-v} below. It will be important to note that the correspondence between functions \Cref{eq:24} and \Cref{eq:4} preserves identical vanishing:

\begin{proposition}\label{pr.vanish}
  Let $Q$ be a quadratic function defined by \Cref{eq:4}. Then, $Q$ vanishes identically on $GF(2^n)$ if and only if \Cref{eq:24} vanishes identically on $GF(2)^n$. 
\end{proposition}
\begin{proof}
  We know from \Cref{th.qq'} that the two functions have the same absolute value for their Walsh transform at $0$. Since the weight is
  \begin{equation*}
    2^{n-1}-\frac{W({\bf 0})}2,
  \end{equation*}
  there are two possibilities:
  \begin{itemize}
  \item the two weights are equal, meaning that they are simultaneously zero or not;    
  \item the two weights add up to $2^n$, in which case it would be impossible for either one of them to be zero: the other one would then be $2^{n-1}$, contradicting the fact that \Cref{eq:24} and $Q$ annihilate the zero vector in $GF(2)^n$ and the zero element of $GF(2^n)$ respectively.
  \end{itemize}
  This finishes the proof.
\end{proof}

\Cref{pr.vanish} allows for relatively simple characterizations of those situations when we do have identical vanishing.

\begin{proposition}\label{pr.when-v}
  Let $Q$ be as in \Cref{eq:4} and $n$ a positive integer. Then, $Q$ vanishes identically on $GF(2^n)$ if and only if one of the following occurs
  \begin{itemize}
  \item the non-zero coefficients $a_i$ in \Cref{eq:4} come in pairs
    \begin{equation*}
      a_i,\ a_{i'},\ i=\pm i'~\mathrm{mod}~n.
    \end{equation*}
  \item same as above, except there is also an additional term $x^{2^i+1}$ with $i=\frac n2~\mathrm{mod}~n$.
  \end{itemize}
\end{proposition}
\begin{proof}
  By \Cref{pr.vanish} we can consider the function \Cref{eq:24} instead, and determine when it can vanish identically on $GF(2)^n$. This happens if and only if it vanishes as a polynomial, i.e. every monomial $x_jx_{j+i}$ appears an even number of times. That this precisely matches the two possibilities in the statement is now immediate.
\end{proof}

It will be convenient to name the following properties appearing in \Cref{pr.when-v}. 

\begin{definition}\label{def.equit}
  A multiset of residues modulo $n$ is {\it equitable} if it can be partitioned into pairs $i$, $i'$ such that
  \begin{equation*}
    i\pm i'=0~\mathrm{mod}~n
  \end{equation*}

  The multiset is {\it semi-equitable} if it is a union of an equitable multiset and an odd number of residues $\frac n2~\mathrm{mod}~n$ (or equivalently, one such residue).  
\end{definition}

\subsection{Exact powers of $x-1$}

The preceding discussion makes it clear that given $Q$ defined by \Cref{eq:4}, it will be important to gain more information about the highest power $(x-1)^{d_Q}$ of $x-1$ dividing
\begin{equation}\label{eq:25}
  A(x) = \sum_{i}a_i(x^i+x^{-i}),\ a_i\in GF(2). 
\end{equation}

Our first remark is

\begin{proposition}\label{pr.odds}
  If all terms of $Q$ have odd subscripts $i,$ then $d_Q \equiv 2~\mathrm{mod}~4$. 
\end{proposition}
\begin{proof} 
  Forming the smallest common denominator in \Cref{eq:25}, the numerator will be
  \begin{equation}\label{eq:26}
    \sum_i a_i (x^{m+i}+x^{m-i})
  \end{equation}
  where $m=\max i$. All exponents in \Cref{eq:26} are even and hence that expression is a square in $GF(2)[x]$. Taking a square root produces a polynomial
  \begin{equation*}
    p(x)\in GF(2)[x]
  \end{equation*}
  which is {\it palindromic} (i.e. the list of coefficients is left-right symmetric), has odd degree and free term $1$.

  The desired conclusion is that the exact power of $x-1$ dividing $p$ has odd exponent. To see this, simply note that if $(x-1)^2=x^2+1$ divides $p$ then the quotient $\frac{p(x)}{x^2+1}$ is again palindromic of odd degree with non-vanishing free term and hence we can proceed by induction on the degree.
\end{proof}

As an immediate consequence we obtain

\begin{corollary}\label{cor.samenu}
  If all subscripts $i$ appearing in the terms of $Q$ have the same $2$-adic valuation $\mu$ then $\nu(d_Q)=\mu+1$. 
\end{corollary}

In general, we can partition the terms of $Q$ (and $A$) according to the $2$-adic valuation $\nu(i)$ of the exponents $i$ in \Cref{eq:25} (i.e. those corresponding to terms with $a_i\ne 0$). We write
\begin{equation*}
  {}_{\mu}Q\quad\text{and}\quad {}_{\mu}A(x)
\end{equation*}
for the partial sums of \Cref{eq:4,eq:25} collecting those terms for which the $2$-adic valuation $\nu(i)$ is $\mu$. We will similarly decorate other objects with left-hand $\mu$ subscripts when needed, indicating an analogous partitioning. For instance, ${}_{\mu}d_Q$ will denote the largest exponent of $x-1$ in ${}_{\mu}A(x)$. 

\begin{corollary}\label{cor.minmu}
  With the above notation and conventions all ${}_{\mu}d_Q$ are distinct and
  \begin{equation*}
    d_Q = \min_{\mu} {}_{\mu}d_Q. 
  \end{equation*}
\end{corollary}
\begin{proof}
  \Cref{cor.samenu} says that ${}_{\mu}d_Q$ has $2$-adic valuation $\mu+1$, hence the conclusion that ${}_{\mu}d_Q$ are distinct. As for the last statement, this is basic polynomial arithmetic, expressing the non-archimedean-ness of the $(x-1)$-valuation on the ring of Laurent polynomials over $GF(2)$: the exact power of $x-1$ dividing a sum of terms with {\it distinct} $(x-1)$-adic valuations is the smallest exact power dividing one of the terms.
\end{proof}

\subsection{\Cref{cj.bal}, even number of nonzero $a_i$}\label{subse.even}

 \Cref{pr.when-v} gives:

\begin{proposition}\label{pr.when-v-ev}
  Define $Q$ by \Cref{eq:4} with an even number of terms and let $n=2^{\nu}$, $\nu\ge 0$. Then, $Q$ vanishes identically on $GF(2^n)$ if and only if its subscripts $i$ form an equitable set modulo $n$ in the sense of \Cref{def.equit}.
\end{proposition}

In particular,

\begin{corollary}\label{cor.all-lower}
  If $Q$ as in \Cref{pr.when-v-ev} vanishes identically on $GF(2^{2^{\nu}})$ then it does on all of its subfields. 
\end{corollary}
\begin{proof}
  Indeed, for the pairs $i,i'$ in the statement of \Cref{pr.when-v-ev} we have $i\pm i'=0$ modulo every $2^{\nu'}$, $\nu'\le \nu$ if we do for $\nu$.
\end{proof}

\begin{proposition}\label{pr.ev-prel}
  Let $Q$ be as in \Cref{eq:4} with an even number of terms. Then, the set of $\nu$ such that $Q_{2^{\nu}}$ is balanced is an interval
  \begin{equation*}
    \nu_{min},\ \nu_{min}+1,\ \cdots,\ \nu_{max},
  \end{equation*}
  possibly empty, with $\nu_{min}>0$ if it exists. 
\end{proposition}
\begin{proof}
  Suppose we do have such $\nu$, i.e. $Q$ is occasionally balanced. By \Cref{cor.no-t} there is a maximal $\nu_{max}$ for which this happens. Now begin traversing the interval
  \begin{equation*}
    0,\ 1,\ \cdots,\ \nu_{max}
  \end{equation*}
  downward.  We know from \Cref{pr.cutoff} that a jump from `balanced' to `unbalanced' entails identical vanishing, and \Cref{cor.all-lower} says that once we encounter such a $\nu$ we have identical vanishing of $Q$ on $GF(2^{2^{\nu'}})$ for all subsequent $\nu'\le \nu$.

  Finally, the fact that $Q$ vanishes on $GF(2)$ follows immediately from the assumption that we have an even number of terms. 
\end{proof}

In order to confirm \Cref{cj.bal} in this case we would have to argue that the interval from \Cref{pr.ev-prel} can only be a singleton or empty. 

\begin{proposition}\label{pr.ev-equit}
Suppose $Q$ has an even number of terms and $Q={}_{\mu}Q$ for some $\mu$. If $2^t\le d_Q$ then the exponents $i$ appearing in $Q$ form an equitable set modulo $2^t$ in the sense of \Cref{def.equit}.   
\end{proposition}
\begin{proof}
  Suppose first that $\mu=0$, i.e. all exponents appearing in $Q$ are odd. In that case we know from \Cref{pr.all-odd} that $Q$ is never balanced. It follows that the restriction of $Q$ to $GF(2^t)$ is identically zero, and hence the set of exponents $i$ is equitable modulo $2^t$ by \Cref{pr.when-v}.

  The general case follows similarly, as we now describe. Once more, by \Cref{pr.0orall,pr.when-v} what we want to show is that $Q_{2^t}$ is not balanced. We abuse notation slightly and work with the RS Boolean function on $GF(2)^{2^t}$ associated to $Q$, denoting it by the same symbol; this will not make a difference by \Cref{th.qq'}. 

  We are assuming all $x_jx_{j+i}$ terms appearing in the expansion of $Q$ have $\nu(i)=\mu$, i.e. the exact power of $2$ dividing all $i$ is $2^\mu$. it follows that $2^{\mu}\le d_Q$, so we may as well assume $2^{\mu}\le 2^t\le d_Q$ (since clearly, if we prove equitability over some large $2^t$ we also prove it for its divisors).

  Let $R$ be the Boolean function associated to the trace function
  \begin{equation*}
    \sum_{i'}\mathrm{Tr}(x^{2^{i'}+1}),
  \end{equation*}
  where each $i'$ is $\frac{i}{2^{\mu}}$ for an $i$ appearing in $Q$ (in other words, we only keep the maximal odd divisors from the $i$s relevant to $Q$). 

  Now make the change of variables
  \begin{equation*}
    x_{r+i2^{\mu}}\leftrightarrow y^{(r)}_{i}
  \end{equation*}
  for all residues $0\le r<2^{\mu}$. Now $Q$ breaks up as a sum of copies of $R$, one for each residue $0\le r<2^{\mu}$ operating on the variables $y^{(r)}$. Since $R$ is unbalanced by the first paragraph of the present proof, so is $Q$.
\end{proof}

We can now prove the even half of \Cref{cj.bal}. 

\begin{theorem}\label{th.ev-fin}
  If $Q$ has an even number of terms then \Cref{cj.bal} holds. 
\end{theorem}
\begin{proof}
  We consider two cases:

  {\bf Case 1: all ${}_{\mu}Q$ have an even number of terms.}  Consider the unique $\mu$ such that $d_Q={}_{\mu}d_Q$ as per \Cref{cor.minmu}. By the definition of $\nu_Q$ we have
  \begin{equation}\label{eq:27}
    2^{\nu_Q-1}<d_Q\le 2^{\nu_Q}
  \end{equation}
  and by \Cref{pr.bet-pow2} the only $n=2^\nu$ over which $Q$ stands a chance of being balanced are those with $\nu\le \nu_Q$.

  Now, $Q$ may or may not be balanced over $2^{\nu_Q}$ itself. As for the strictly smaller powers of two $n=2^{\nu}$, $\nu<\nu_Q$, \Cref{eq:27} implies that they fall under the scope of \Cref{pr.ev-equit} and hence the exponents of every ${}_{\mu}Q$ form an equitable set modulo $2^{\nu}$. By \Cref{pr.when-v-ev} every ${}_{\mu}Q_n$ vanishes, and hence so does $Q_n$ (thus failing to be balanced).

  {\bf Case 2: general.}  Now suppose there is some $\mu$ such that ${}_{\mu}Q$ has an odd number of terms and let $\mu_o$ (for `odd') be the smallest such $\mu$. We then have
  \begin{equation}\label{eq:29}
    d_Q = \min_{\mu\le \mu_o} {}_{\mu}d_Q. 
  \end{equation}
Indeed, since ${}_{\mu_o}Q$ has an odd number of terms, ${}_{\mu_o}d_Q=2^{\mu_o+1}$. On the other hand, for {\it larger} $\mu>\mu_o$ we have
  \begin{equation*}
    {}_{\mu}d_Q\ge 2^{\mu+1}>2^{\mu_o+1}={}_{\mu_o}d_Q
  \end{equation*}
  and hence \Cref{eq:29} follows from \Cref{cor.minmu}.

  \Cref{eq:29} implies that $d_Q$ is achieved as
  \begin{equation*}
    {}_{\mu}d_Q\le {}_{\mu_o}d_Q=2^{\mu_o+1}
  \end{equation*}
  for a unique $\mu\le \mu_o$. By \Cref{pr.bet-pow2}, the only powers of two $n=2^{\nu}$ for which $Q_n$ can be balanced are those with $\nu\le \nu_Q\le \mu_o+1$.

  As in the proof of Case 1, $Q$ may or may not be balanced over $2^{\nu_Q}$. Smaller powers of two $2^{\nu}$, $\nu<\nu_Q$ are all dominated by every ${}_{\mu}d_Q$, $\mu\le \mu_o$ and hence all ${}_{\mu}d_Q$, $\mu<\mu_o$ have equitable sets of exponents modulo $2^{\nu}$.

  On the other hand, all $\mu\ge \mu_o$ dominate $\nu<\nu_Q\le \mu_o+1$ and hence the exponents of ${}_{\mu}Q$, $\mu\ge \mu_o$ are all zero modulo $2^{\nu}$. Since the number of such exponents is even, the exponent set of $Q$ as a whole is equitable modulo $2^\nu$. This makes $Q_{2^{\nu}}$ unbalanced, concluding the proof.
\end{proof}

The following example illustrates \Cref{cj.bal} in the small case where there are two nonzero terms.
\begin{example} \label{exev}
Suppose $a_i \neq 0$ for $i= 3$ and $4.$ Thus $Q_n(x)$ in \Cref{eq:4} corresponds to $f_n= (0,3)_n+(0,4)_n$ when $n \geq 5.$ The algorithm of \cite{C18} shows that the recursion for the weights of $f_n,~n \geq 9,$ has order $15$.  If we let $u(n)$ denote the $n$-th term of the recursion sequence with $u(n)= wt(f_n)$ for $n \geq 9,$ then the recursion is
$$u(n)=2u(n-1)+16u(n-7)-32u(n-8)-128u(n-14)+256u(n-15)$$
for $n \geq 9$ (we begin at $n=9$ to avoid the short function $(0,4)_8;$ see the first paragraph in \Cref{genquad}).  Computation of $wt(f_n)$ for $9 \leq n \leq 23$ enables all of the terms $u(n)$ to be computed; in particular, the terms with $1 \leq n \leq 8$ can be found by extending the recursion backwards from $u(9).$ This gives the following initial segment $\{u(n): 1 \leq n \leq 18\}:$
$$\{0, 2, 6, 12, 20, 32, 0, 112, 240, 512, 1056, 2112, 4160, 8192, 16 256, 32 512,65 280, 131 072\}.$$
We see that the interval such that $Q_{\nu}$ is balanced  in \Cref{pr.ev-prel} is the single integer $1$ since $u(n) = 2^{n-1}$ (balanced) for $n=2$ but $u(4)=12.$ Thus $d(Q)=2$ in \Cref{cj.bal} and so the balanced functions $f_n$ are precisely those with $n \equiv 2 \bmod 4.$ 

The actual weight $wt(f_8)$ is $136$, but the recursion value $u(8)=112=wt((0,3)_8)$ is not equal to the weight. This is because \Cref{eq:24} reduces to the function $(0,3)_8$ since the short function $(0,4)_8$ vanishes identically in \Cref{eq:24}. Similarly $wt(f_6)=24$ but $u(6)=32=wt((0,4)_6)$ because the function $(0,3)_6$ is short. Also $u(7)=0$ since the function $f_7$ vanishes identically because $(0,3)_7 = (0,4)_7.$ We do obtain $u(5)=wt(f_5)=20,$ even tho $n<9$ in this instance.   
 \end{example}

\subsection{\Cref{cj.bal}, odd number of nonzero $a_i$}\label{subse.odd}

We now specialize to the titular case. The branch of \Cref{pr.when-v} valid here is

\begin{proposition}\label{pr.when-v-odd}
  Define $Q$ by \Cref{eq:4} with an odd number of terms and let $n=2^{\nu}$, $\nu\ge 0$. Then, $Q$ vanishes identically on $GF(2^n)$ if and only if its subscripts $i$ form a semi-equitable set modulo $n$ in the sense of \Cref{def.equit}.
\end{proposition}
\begin{proof}
As observed, this is simply the variant of \Cref{pr.when-v} applicable here, given that we are assuming an odd number of terms. 
\end{proof}

This allows us to considerably narrow down the possibilities for when $Q$ is balanced.

\begin{proposition}\label{pr.odd-prel}
  Let $Q$ be as in \Cref{eq:4} with an odd number of terms. Then, the set of $\nu$ such that $Q_{2^{\nu}}$ is balanced is non-empty and takes one of these forms:
  \begin{enumerate}[(a)]
  \item\label{item:1} an initial segment
    \begin{equation*}
      \{0,\ 1,\ \cdots,\ \nu_{max}\}
    \end{equation*}
    of $\mathbb{Z}_{\ge 0}$;
  \item\label{item:2} an initial segment with one missing element.
  \end{enumerate}
\end{proposition}
\begin{proof}
  Non-emptiness follows from the odd number of terms assumption and \Cref{cor.odd-n}.
  
  We know from \Cref{cor.no-t} that there is a largest value $\nu_{max}$ such that $Q_{2^{\nu_{max}}}$ is balanced. Now note that there is at most one $\nu$ such that $Q_{2^{\nu}}$ vanishes identically on $GF(2^{2^{\nu}})$: if such a $\nu$ exists then the paired-up $i$ and $i'$ in \Cref{pr.when-v-odd} will still add up or subtract to zero modulo any smaller power of two, while the single index $i$ with
  \begin{equation*}
    i=2^{\nu-1}~\mathrm{mod}~2^{\nu}
  \end{equation*}
  can only satisfy that modular congruence for a single $\nu$, and hence no smaller $\nu'<\nu$ qualify.  

  In conclusion, in traversing the interval
  \begin{equation*}
    0,\ 1,\ \cdots,\ \nu_{max}
  \end{equation*}
  downward we can only make a transition from `balanced' to `unbalanced' at most once.
\end{proof}

Verifying \Cref{cj.bal} in the odd-number-of-terms case entails eliminating option \labelcref{item:2} in \Cref{pr.odd-prel}.

\begin{theorem}\label{th.od-fin}
  If $Q$ has an odd number of terms then \Cref{cj.bal} holds.
\end{theorem}
\begin{proof}
  The argument will be very similar to that in \Cref{th.ev-fin}, and can in fact be replicated virtually verbatim. This time around, since $Q$ has an odd number of terms there must be some ${}_{\mu}Q$ with the same property and hence we can choose a smallest $\mu_o$ as before.

  The single point of divergence between the proof of \Cref{th.ev-fin} and the present one occurs in the very last paragraph of the former: here, the number of exponents of ${}_{\mu}Q$, $\mu\ge \mu_o$ is odd rather than even. This means that for $\nu<\nu_Q$ the set of exponents of $Q$ fails to be semi-equitable, making $Q_{2^{\nu}}$ balanced by \Cref{pr.when-v-odd}.

  In conclusion, $Q$ will be balanced precisely for $2^{\nu}$ for $\nu$ ranging over an initial interval of non-negative integers, as desired.
\end{proof}

The following example illustrates \Cref{cj.bal} in the simplest case for which there are three nonzero terms.
\begin{example} \label{exodd}
Suppose $a_i \neq 0$ for $i= 1, 2, 3.$ Thus $Q_n(x)$ in \Cref{eq:4} corresponds to $f_n= (0,1)_n+(0,2)_n+(0,3)_n$ when $n \geq 4.$ The algorithm of \cite{C18} shows that the recursion for the weights of $f_n,~n \geq 7,$ has order $9$.  If we let $u(n)$ denote the $n$-th term of the recursion sequence with $u(n)= wt(f_n)$ for $n \geq 7,$ then the recursion is
$$u(n)=2u(n-1)+4u(n-4)-8u(n-5)-16u(n-8)+32u(n-9)$$
for $n \geq 7$ (we begin at $n=7$ to avoid the short function $(0,3)_6;$ see the first paragraph in \Cref{genquad}).  Computation of $wt(f_n)$ for $7 \leq n \leq 21$ enables all of the terms $u(n)$ to be computed; in particular, the terms with $1 \leq n \leq 6$ can be found by extending the recursion backwards from $u(7).$ This gives the following initial segment $\{u(n); 1 \leq n \leq 16\}:$
$$\{1, 2, 4, 0, 16, 32, 64, 160, 256, 512, 1024, 2304, 4096, 8192, 16384,33280\}.$$
We see that the initial segment in \Cref{pr.odd-prel} is $\{0,1\},$ since $u(n) = 2^{n-1}$ (balanced) for $n=1,2$ but $u(4)=0.$  The actual weight $wt(f_4)$ is $6$ (since $f_4= (0,2)_4$ is a short bent function), but the recursion value $u(4)$ is not equal to the weight.  We do obtain $u(5)=wt(f_5)=16,$ even tho
 $n<7$ in this instance.   Note that $Q_4=0$ on $GF(2^4),$ so  \Cref{pr.when-v-odd} applies.  Indeed the set of subscripts is $\{1,2,3\}$ and this is a 
 semi-equitable set (\Cref{def.equit}) since $3+1 \equiv 0~\mathrm{mod}~4$ and $2 \equiv~\frac{4}{2}~  \mathrm{mod}~4.$
\end{example}

\section{Future work} \label{fut}

One important project is to extend the complete description of the affine equivalence classes for the MRS quadratic functions to the case of general RS functions. Because of \Cref{quadeq}, the results above can be used to decide whether two given quadratic RS functions are affine equivalent by means of a straightforward calculation of their nonlinearity and weight.  Obtaining a count of the equivalence classes (like \Cref{MRSeq} for the MRS quadratic functions)  for general RS functions in $n$ variables seems to require new ideas.

Given the algebraic normal form (ANF) of a RS Boolean function $f_n$ in $n$ variables with degree $d,$ the method described in \cite{C17, C18} for finding a linear recursion for $wt(f_n)$ gives the correct values for the Hamming weights only if $n$ is taken large enough to avoid including any short and bent functions in the calculation.  

If $f_n$ is quadratic, then we know from \Cref{n=2d+v} that $f_n$ is bent if and only if $v(n)=0,$ so it is easy to specify $B(f) = B$ such that $n \geq B$ is large enough.  Also, the ANF does not make sense if $n < d,$ but the weight recursion can be extended backwards to give values for any $n \geq 1.$  Given $f_n$ in the trace form $Q_n(x)$   (as defined in \Cref{se.prel}), we define the recursion values  for all $n \geq 1$ whatever the degree $d$ is.  In particular, \Cref{cj.bal} answers the question of when $Q_n(x)$ is balanced (that is, the corresponding recursion ``weight'' is $2^{n-1}$) for all 
$n \geq 1$ without worrying about the bent functions. Also, \Cref{cj.bal} shows that determining the ``weights'' for $1 \leq n < d$ is easy.  The reason why the bent functions are not a concern is that they disappear in the trace computations, since the only monomial bent quadratic RS functions  are the short ones $(0,t)_{2t},$ which are identically $0$ functions if their ANF is not reduced to $t$ monomials instead of $2t$ (see \Cref{pr.0orall}, \Cref{pr.vanish} and the discussion in between). These facts suggest that the algorithm as described in \cite{C17}, at least in the quadratic case, could be simplified by omitting any monomial bent function terms from the calculations.  

Computation of the roots of the recursion polynomials for various quadratic RS functions $f_n$   suggests that the roots of these polynomials are always algebraic integers with absolute value $\sqrt{2}.$   Neither the methods of \cite{C17,C18} nor the results in the present paper seem able to give any insight into this conjecture, but by using some new ideas we shall prove it and much more in a later paper \cite{CC20}. 

Another very interesting question is whether results similar to \Cref{cj.bal} are true for balanced functions of higher degree.


\end{document}